\documentclass[11pt,reqno,tikz]{amsart}
\usepackage{amsmath,amsthm,amssymb,bm,bbm,dsfont,braket}
\usepackage[mathscr]{eucal}
\usepackage{amsaddr}
\usepackage{upgreek}
\usepackage[left=2cm,right=2cm,top=2cm,bottom=2cm]{geometry}
\usepackage{mathtools}\mathtoolsset{centercolon}
\usepackage{cite}

\usepackage[shortlabels]{enumitem}
\usepackage{setspace}
\usepackage{graphicx}
\usepackage{verbatim}
\usepackage{float}
\usepackage{placeins}
\usepackage{array}
\usepackage{booktabs}
\usepackage{multicol,multirow,tabularx}
\usepackage{threeparttable}
\usepackage[update,prepend]{epstopdf}
\usepackage{amsfonts,amssymb,dsfont,xfrac,comment}
\usepackage[abs]{overpic}

\makeatletter
\def\adl@drawiv#1#2#3{%
	\hskip0
	\tabcolsep
	\xleaders#3{#2 0\@tempdimb #1{1}#2 0.5\@tempdimb}%
	#2\z@ plus1fil minus1fil\relax
	\hskip0\tabcolsep}
\newcommand{\cdashlinelr}[1]{%
	\noalign{\vskip\aboverulesep
		\global\let\@dashdrawstore\adl@draw
		\global\let\adl@draw\adl@drawiv}
	\cdashline{#1}
	\noalign{\global\let\adl@draw\@dashdrawstore
		\vskip\belowrulesep}}
\makeatother

\usepackage[usenames,dvipsnames]{color}
\usepackage{colortbl}
\usepackage{arydshln}
\usepackage[hidelinks,linktocpage=true]{hyperref}
\hypersetup{
	unicode=false,          
	pdftoolbar=true,        
	pdfmenubar=true,        
	pdffitwindow=false,     
	pdfstartview={FitH},    
	pdftitle={My title},    
	pdfauthor={Author},     
	pdfsubject={Subject},   
	pdfcreator={Creator},   
	pdfproducer={Producer}, 
	pdfkeywords={keyword1} {key2} {key3}, 
	pdfnewwindow=true,      
	colorlinks=true,        
	linkcolor=Red,          
	citecolor=ForestGreen,  
	filecolor=Magenta,      
	urlcolor=BlueViolet,    
}
\usepackage{doi}
\usepackage{url}
\usepackage{caption, subcaption}
\usepackage{enumitem}
\usepackage{shadethm}
\usepackage{tikz,pgf}

\usepackage{cleveref}

\makeatletter
\ifx\@NODS\undefined%

\let\mathbb=\mathds
\else%
\fi
\makeatother

\newcommand{\Exterior}{\mathchoice{{\textstyle\bigwedge}}%
	{{\bigwedge}}%
	{{\textstyle\wedge}}%
	{{\scriptstyle\wedge}}}

\DeclareMathOperator*{\argmax}{arg\,max}
\DeclareMathOperator*{\argmin}{arg\,min} 

\DeclareMathOperator{\Tr}{Tr}

\DeclareMathOperator{\e}{\mathrm{e}}

\def\X{\mathsf{X}}
\def\Y{\mathsf{Y}}
\def\A{\mathsf{A}}
\def\B{\mathsf{B}}
\def\C{\mathsf{C}}
\def\U{\mathsf{U}}
\def\V{\mathsf{V}}
\def\R{\mathsf{R}}

\newcommand{\be}{{\mathbf e}}

        \def\cB{{\mathcal B}}

\def\0{{\mathbf{0}}}
\def\1{{\mathbf{1}}}
\def\2{{\mathbf{2}}}
\def\3{{\mathbf{3}}}
\def\4{{\mathbf{4}}}
\def\5{{\mathbf{5}}}
\def\6{{\mathbf{6}}}

\def\7{{\mathbf{7}}}
\def\8{{\mathbf{8}}}
\def\9{{\mathbf{9}}}

\def\be{\begin{equation}}
\def\ee{\end{equation}}
\def\bea{\begin{eqnarray}}
\def\eea{\end{eqnarray}}

\def\eps{\varepsilon}

\theoremstyle{plain}
\newshadetheorem{theo}{Theorem}
\newshadetheorem{prop}{Proposition} 
\newshadetheorem{lemm}{Lemma} 
\newtheorem*{fact}{Fact} 
\newshadetheorem{theorem}{Theorem}

\theoremstyle{definition}
\newtheorem{defn}{Definition} 

\theoremstyle{remark}
\newtheorem{remark}{Remark}

\makeatletter
\newcommand{\opnorm}{\@ifstar\@opnorms\@opnorm}
\newcommand{\@opnorms}[1]{%
	$\left|\mkern-1.5mu\left|\mkern-1.5mu\left|
	#1
	\right|\mkern-1.5mu\right|\mkern-1.5mu\right|$
}
\newcommand{\@opnorm}[2][]{%
	\mathopen{#1|\mkern-1.5mu#1|\mkern-1.5mu#1|}
	#2
	\mathclose{#1|\mkern-1.5mu#1|\mkern-1.5mu#1|}
}
\makeatother

\usepackage{tikz}
\usetikzlibrary{arrows.meta} 
\tikzset{>={Latex[length=4,width=4]}} 
\usetikzlibrary{calc}

\colorlet{mylightblue}{blue!5!white}
\colorlet{mydarkblue}{blue!30!black}
\colorlet{myblue}{blue!50!black}
\colorlet{myred}{red!50!black}
\colorlet{mydarkred}{red!30!black}
\colorlet{mydarkgreen}{green!30!black}

\newcommand{\sh}{\kern-0.08em$^\textbf{\#}$\hspace{-3pt}}
\renewcommand{\b}{\kern-0.06em$\flat$}

\begin{document}

\let\origmaketitle\maketitle
\def\maketitle{
	\begingroup
	\def\uppercasenonmath##1{} 
	\let\MakeUppercase\relax 
	\origmaketitle
	\endgroup
}

\title{\bfseries \Large{ 
Simple and Tighter Derivation of Achievability for \\ Classical Communication over Quantum Channels
		}}

\author{ \normalsize \textsc{Hao-Chung Cheng}}
\address{\small  	
Department of Electrical Engineering and Graduate Institute of Communication Engineering,\\ National Taiwan University, Taipei 106, Taiwan (R.O.C.)\\
Department of Mathematics, National Taiwan University\\
Center for Quantum Science and Engineering,  National Taiwan University\\
Hon Hai (Foxconn) Quantum Computing Center, New Taipei City 236, Taiwan (R.O.C.)\\
Physics Division, National Center for Theoretical Sciences, Taipei 10617, Taiwan (R.O.C.)\\
}

\email{\href{mailto:haochung.ch@gmail.com}{haochung.ch@gmail.com}}

\date{\today}
\begin{abstract}
Achievability in information theory refers to demonstrating a coding strategy that accomplishes a prescribed performance benchmark for the underlying task.
In quantum information theory, the crafted Hayashi--Nagaoka operator inequality is an essential technique in proving a wealth of one-shot achievability bounds since it effectively resembles a union bound in various problems.
In this work, we show that the so-called pretty-good measurement naturally plays a role as the union bound as well. A judicious application of it considerably simplifies the derivation of one-shot achievability for classical-quantum channel coding via an elegant three-line proof.

The proposed analysis enjoys the following favorable features.
(i) The established one-shot bound admits a closed-form expression as in the celebrated Holevo--Helstrom theorem.
Namely, the average error probability of sending $M$ messages through a classical-quantum channel is upper bounded by the minimum error of distinguishing the joint channel input-output state against $(M-1)$ decoupled product states.
(ii) Our bound directly yields asymptotic achievability results in the large deviation, small deviation, and moderate deviation regimes in a unified manner.
(iii) The coefficients incurred in applying the Hayashi--Nagaoka operator inequality or the quantum union bound are no longer needed. 
Hence, the derived one-shot bound sharpens existing results relying on the Hayashi--Nagaoka operator inequality.
In particular, we obtain the tightest achievable $\varepsilon$-one-shot capacity for classical communication over quantum channels heretofore, improving the third-order coding rate in the asymptotic scenario.
(iv) Our result holds for infinite-dimensional Hilbert space.
(v) The proposed method applies to deriving one-shot achievability for classical data compression with quantum side information, entanglement-assisted classical communication over quantum channels, and various quantum network information-processing protocols.
\end{abstract}

\maketitle
\vspace{-2.5em}
\tableofcontents

\section{Introduction} \label{sec:introduction}

Communicating classical information over a noisy quantum channel is a foundational task in quantum information science. 
To protect the transmitted messages against potential noise, an indispensable coding strategy is employed.
At the transmitter, Alice initiates the procedure by encoding each message $m \in \{1,2,\ldots, M\}$ into an $n$-qubit quantum state.
Suppose in the communication process that each qubit suffers from independent and identically distributed (i.i.d.) quantum noise, which is characterized by an i.i.d.~quantum channel.
Then, at the receiver Bob performs a quantum measurement on the corrupted quantum system to extract the decoded message $\hat{m}$.

Via a coding strategy based on the so-called \emph{quantum typicality}, the well-known Holevo--Schumacher--Westmoreland (HSW) theorem \cite{JS94, HJS+96, SW97, Hol98, Win99b, Win99, Wilde2} states that the probability of erroneous decoding, $\varepsilon := \Pr\left\{\hat{m} \neq m\right\}$, vanishes asymptotically in the limit of $n\to \infty$, whenever the number of bits to be sent per qubit ($\lim_{n\to \infty} \frac1n \log M$) is below the \emph{channel capacity}.

The HSW theorem extends the seminal work of Shannon \cite{Sha48} to the quantum scenario, and hence, it is one of the fundamental core stones in quantum information theory.
However, the HSW coding strategy relies on certain technical assumptions that could be physically demanding.
First, the asymptotically large qubit number $n$ requires the quantum devices to implement arbitrarily large encoding and decoding.
Second, the actual quantum noise may be correlated among several qubit systems; hence, the underlying quantum noises are not independent.
Third, even if the noises are independent, they may not be stationary; namely, the noise acting on the first qubit is not identical to that on the last qubit.

To circumvent the aforementioned  technical requirements, \emph{one-shot} quantum information theory emerges as a new research stream to consider the scenario that no structural hypotheses are imposed on the underlying quantum state or channel. 
The ultimate goal is to characterize the \emph{optimal trade-off} between the error probability $\varepsilon$ and the message size $M$ of transmission, for which the channel is used only once.
Such a study allows us to better understand the fundamental capability of  one-shot communication. Therefore it may serve as a general guideline for designing the next-generation quantum information-processing systems.
However, without the i.i.d.~repetitions of channel use, conventional methods based on quantum typicality are no longer applicable.
Hence more refined and sophisticated coding techniques are requisite for the one-shot analysis\footnote{
	To the best of our knowledge, the first achievability analysis for one-shot coding in quantum information theory was made by Burnashev and Holevo \cite{BH98} in studying pure-state classical-quantum (c-q) channels (based on a Gram matrix technique), and a one-shot strong converse for c-q channel coding was proved by Ogawa and Nagaoka \cite{ON99} via extending Arimoto's approach \cite{Ari76} to the noncommutative scenario.
	Technically speaking, the quantum dense coding \cite{BW92} and teleportation \cite{BBC+93} are both one-shot protocols---however, the original references only concerned noiseless channels. 
	Further, the early work of quantum state discrimination \cite{Hel67, YKL75} also considered the one-shot setting.
	This paper will focus on the \emph{packing-type problems} in one-shot quantum information theory with general (noisy) quantum channels without specific constraints.
	We skip the \emph{covering-type problems} such as quantum covering \cite{AW02, Hay06, Hay15, Hay17, ADJ17, AJW19a,	CG22, SGC22b}, privacy amplification \cite{Ren05, TSS+11, Hay06, Hay13, Tom16, Dup21, KL21, SGC22a, SGC22b}, decoupling \cite{Dup10, DBW+14, LY21a}, and simulation \cite{6757002, luo_channel_2006,  berta_quantum_2013, LY21b} since their achievability techniques are different from that of the packing problems.
}.

Why is it challenging to design and analyze good coding strategies for a one-shot quantum information-processing task?
Essentially, a proper coding scheme aims to enforce the error probability, $\Pr\{ \bigcup_{\hat{m}\neq m} \mathcal{E}_{\hat{m}\mid m}  \}$, for sending each message $m$ small. 
Here, $\mathcal{E}_{\hat{m}\mid m}$ denotes the event of decoding to $\hat{m}$ when message $m$ was sent.
Yet, the analysis and computational evaluation of such a union error event for a nontrivial quantum measurement could be quite difficult (even in the classical scenario).
A useful trick in this effort is the following \emph{union bound}:
\begin{align} \label{eq:union}
	\Pr\left\{ \bigcup\nolimits_{\hat{m}\neq m} \mathcal{E}_{\hat{m}\mid m}  \right\}
	\leq \sum\nolimits_{\hat{m}\neq m} \Pr \left\{ \mathcal{E}_{\hat{m}\mid m}  \right\}.
\end{align}
In view of the right-hand side of Eq.~\eqref{eq:union}, 
the decoding rule remains to minimize $(M-1)$ pairwise error probabilities of deciding $m$ against each $\hat{m}\neq m$.
This then serves as the general principle of coding design.

The above coding strategy has achieved prevailing success in classical information theory, channel coding, and modern communication systems \cite{Fei55, Fan61, Str62, Gal63, Gal65, Gal68, VH94, Han03, SS06, Gal08, RU08, Hay09b, PPV10, Lap17}. 
However, the quantum union bound of the form \eqref{eq:union} is highly nontrivial due to the noncommutative nature of quantum mechanics \cite{Aar06, GLL12, Sen12, Wil13, Gao15, OMW19}.
The first attempt to design a good one-shot coding scheme for general classical-quantum (c-q) channels (without the i.i.d~condition) was proposed by Hayashi and Nagaoka, in which, a powerful operator inequality \cite[Lemma 2]{HN03} was proved: for any  positive semi-definite operators $0\leq A\leq \mathds{1}$ and $B\geq 0$,
\begin{align} \label{eq:HN03}
	\mathds{1} - \frac{A}{A+B} \leq (1+c) (I-A) + (2+c+c^{-1}) B, \quad \forall c>0,
\end{align}
where we denote a \emph{noncommutative quotient} by 
\begin{align} \label{eq:quotient}
	\frac{A}{B}:= B^{-{1}/{2}} A B^{-{1}/{2}}\,
\end{align} 
(here, the inverse is defined only on the support of the operator in the denominator).
At the first glimpse of Eq.~\eqref{eq:HN03}, it is not obvious how it resembles the union bound as Eq.~\eqref{eq:union} and how 
it is applied in analyzing the error probability in channel coding; nonetheless, an ingenious application of it by Hayashi and Nagaoka \cite{HN03} yields a Feinstein-type bound for achieving the c-q channel capacity \cite{Fei54}, \cite[Theorem 1]{VH94}.
Later, Oskouei, Mancini and Wilde proposed a quantum union bound with similar coefficients as in Eq.~\eqref{eq:HN03}, and hence achieved the same error bound as Refs.~\cite{HN03}, \cite{OMW19} \footnote{We refer the reader to Ref.~\cite[Theorem 2.1]{OMW19} for the precise statement of the quantum union bound. An application of it with the sequential decoding strategy \cite{Wil13} leads to the same achievability bound as in \cite{HN03, WR13} (for infinite-dimensional Hilbert space as well) \cite[Theorem 5.1]{OMW19}.}.
Subsequently, Hayashi and Nagaoka's analysis based on Eq.~\eqref{eq:HN03} lays a technical cornerstone in a wealth of one-shot and asymptotic achievability results in  quantum information theory, wherein a quantum measurement for extracting classical information from a quantum system is needed.
For example, letting the coefficient $c$ in Eq.~\eqref{eq:HN03} be a fixed constant along with a \emph{quantum Chernoff bound} \cite{ACM+07, ANS+08, Hay07, JOP+12} delivers a large deviation bound for c-q channel coding \cite{Hay07}.
Letting $c = \sfrac{1}{\sqrt{n}}$ for an $n$-fold i.i.d.~repetition of a c-q channel, Eq.~\eqref{eq:HN03} achieves the second-order coding rate \cite{PPV10, Tan14, TH13, Li14, TV15} in the small deviation regime. Later, both results were extended to the moderate deviation regime \cite{CH17, CTT2017} accordingly. In addition, Anshu, Jain, and Warsi proposed a \emph{position-based coding} for achieving entanglement-assisted classical communication over quantum channels \cite{AJW19a}, which also relies on the Hayashi--Nagaoka operator inequality in Eq.~\eqref{eq:HN03}.

Apart from the success and significance of Hayashi and Nagaoka's approach, there are still conceptual and practical subtleties.
First, the technically sophisticated proof of the operator inequality \eqref{eq:HN03} may blind the insight of the analysis and hide the reason why such a coding strategy works.
Does there exist a good quantum coding strategy that naturally reflects the union bound as in Eq.~\eqref{eq:union} so that the analysis is more interpretable?
Second, is it possible to tighten the one-shot achievability bound for quantum information-theoretic tasks by eliminating the incurred coefficients in terms of $c$ \footnote{
	It is generally believed that the coefficients in Eq.~\eqref{eq:HN03} or in the quantum union bound cannot be removed \cite{OMW19}. However, achieving a tighter one-shot bound in quantum information theory without such coefficients is possible. That is the keynote of the present paper.
}  \footnote{A one-shot achievability bound without the coefficients in terms of $c$ was actually proved by Beigi and Gohari \cite{BG14}
	but the comparison to Hayashi and Nagaoka's bound \cite[Lemma 3]{HN03}, \cite[Theorem 1]{WR13} was missing.
	We will compare Beigi and Gohari's result \cite[Corollary 1]{BG14} with ours (Theorem~\ref{theo:main}) in Section~\ref{sec:comparison}.
}?
Removing those coefficients may seem superficial.		
However, we remark that every bit in an analytical bound counts in the one-shot setting; one cannot ignore any coefficient.
On top of that, the unnecessary coefficients may often \emph{trivialize} the $(\varepsilon, M)$ trade-off.
For instance, existing one-shot bounds on the error probability $\varepsilon$ could trivially be greater than $1$ for $\log M$ close to the channel capacity.
This analysis then provides no useful characterizations of certain system configurations for practical communication.
Lastly, the Hayashi--Nagaoka decoder \cite{HN03, WR13} involves solving a positive semi-definite program (SDP) to obtain the mathematical description
of quantum measurement, for which the computational complexity is exponential in the number of qubits.
Moreover, a quantum algorithm for implementing the Hayashi--Nagaoka decoder is still missing.

In this paper, we give affirmative answers to the above concerns and questions by showing that the so-called \emph{pretty good measurement} (PGM) \cite{Bel75, HW94} naturally plays a role as the union bound.
Together with the random coding technique, it yields a one-shot achievability bound via a much simpler and self-explainable analysis, which is merely based on previously known facts. 
The coefficients mentioned above in terms of $c$ are not required anymore, and hence, the established one-shot bound is sharpened. Furthermore, the proof itself provides a more transparent connection between c-q channel coding and binary quantum hypothesis testing.

To present our result, we first introduce a \emph{noncommutative minimal} between two positive semi-definite operators $A$ and $B$ as \footnote{The noncommutative minimal is indeed unique. We refer the reader to Section~\ref{sec:minimal} for a more precise statement.}
\begin{align} \notag
	A\wedge B := \frac{A+B-|A-B|}{2}.
\end{align}
This quantity is prominent in quantum state discrimination since the celebrated \emph{Holevo--Helstrom theorem} \cite{Hel67, Hol72, Hol76} endowed it with an operational meaning \footnote{
	Note here that such an operational interpretation in quantum statistical decision theory applies to general self-adjoint positive semi-definite operators $A$ and $B$ whose traces are not necessarily normalized.
	One can think of that the operators $A$ and $B$ already incorporate the prior probabilities and the penalty for the erroneous decision; see e.g.~\cite[\S 2.2.2]{Hol01} and \cite[Eq.~(1)]{Li16}.
}:
\[
\textit{\,$\Tr\left[ A\wedge B \right]$ determines the minimum `\emph{error}' of discrimination between operators $A$ and $B$.
}
\]

The coding strategy for a c-q channel $x\mapsto \rho_{\B}^x$ (which maps each classical symbol $x$ to a density operator $\rho_{\B}^x$) proceeds as follows. 
The encoding is via a random codebook $\left\{x(1),x(2),\ldots, x(M) \right\}$, in which each codeword $x(m)$ is drawn pairwise independently according to an arbitrary probability distribution $p_{\X}$.
The decoding is via the PGM with respect to the corresponding channel output states \cite{Bel75, HW94}:
\begin{align} \notag
	\left\{ \frac{\rho_{\B}^{x(m)}}{ \sum_{\bar{m}=1}^M \rho_{\B}^{x(\bar{m})}  } \right\}_{m=1}^M.
\end{align}
We show that the associated average error probability is upper bounded by (Theorem~\ref{theo:main}):
\begin{align}
	\Tr\left[ \rho_{\X\B} \wedge (M-1) \rho_{\X}\otimes \rho_{\B} \right], \label{eq:3}
\end{align}
with $\rho_{\X\B} = \sum_{x\in\X} p_{\X} (x)|x\rangle\langle x|\otimes \rho_{\B}^x$ the resulting joint bipartite state between the channel input and output.
Via the Holevo--Helstrom theorem, the established bound in Eq.~\eqref{eq:3} provides us the following interpretation for sending $M$ messages over a c-q channel
(see Section~\ref{sec:Main} for the detailed explanation)
:
\[
\textit{The average error probability is upper bounded by the error of
	discriminating $\rho_{\X\B}$ and $(M\!-\!1)\rho_{\X}\otimes \rho_{\B}$.}
\]

Below, let us elaborate on the intuition of the proposed coding strategy and why PGM works well.
The key observation is that using the PGM to discriminate $M$ states at the channel output is {exactly} equivalent to (an average of) binary discrimination between each channel output state, say, e.g.~$\rho_{\B}^{x(m)}$ against the remaining $(M-1)$ states $\rho_{\B}^{x(\bar{m})}$ for all $\bar{m}\neq m$ using a two-outcome PGM.
In this regard, PGM works as a \emph{one-versus-rest} classification strategy; see Figure~\ref{fig:simple1}.
Most importantly, this manifests the fact that PGM effectively resembles the \emph{quantum union bound} as shown in the right-hand side of Eq.~\eqref{eq:union}.
By taking the conditional expectation $\mathds{E}_{x(\bar{m}) \mid x(m) }$ over the random codebook, the remaining states are hence \emph{averaged} to $(M-1)$ identical marginal states $\rho_{\B}$; see Figure~\ref{fig:simple2}.
	After taking expectation $\mathds{E}_{x(m)}$, the error bound is equivalent to discriminating the joint state $\rho_{\X\B}$ between channel input and output against $(M-1)$ product states $\rho_{\X} \otimes \rho_{\B}$ as shown in Figure~\ref{fig:simple3}.
This gives the elegant and clean bound in Eq.~\eqref{eq:3}.

\medskip
The proposed simple derivation enjoys the following favorable features.
\begin{enumerate}[(I)]
	\item The one-shot achievability bound in Eq.~\eqref{eq:3} admits a closed-form expression as the Holevo--Helstrom theorem. Computing such a bound is more time efficient than the previous results in terms of entropic quantities involving optimizations (see Remark~\ref{remark:time} in Section~\ref{sec:Main}). 
	
	\item 
	The proposed coding scheme based on the pretty-good measurement is directly implementable via the existing quantum algorithm by Gily{\'{e}}n \textit{et al.} \cite{GLM+22}.
	
	\item The self-explainable proof signifies a more lucid connection between c-q channel coding and hypothesis testing.
	Moreover, our coding strategy and analysis show that PGM effectively works as a union bound by itself.
	Hence, neither the operator inequality \eqref{eq:HN03} nor a quantum union bound is needed.
	
	
	\item The proposed bound in Eq.~\eqref{eq:3} is free of parameter $c$ as in Eq.~\eqref{eq:HN03}. This then shows that the established one-shot achievable error bound is tighter than previously known results based on the Hayashi--Nagaoka operator inequality, Eq.~\eqref{eq:HN03}; see Section~\ref{sec:comparison} and Table~\ref{table:comparison} therein for a comparison with existing results.
	Moreover, it unifies asymptotic derivations in the large, small, and moderate derivation regimes.
	We refer the reader to Figure~\ref{fig:flow} for a schematic flow chart.
	
	\item The proposed analysis applies to infinite-dimensional quantum systems \cite{Ser17} as well, e.g.~communication over infinite-dimensional channels with energy constraints or cost constraints \cite{Hol02, Hol12}.
	
	\item The proposed methods via pretty-good measurement naturally extend to various quantum information-theoretic tasks, leading to more profound and sharpened results. These tasks include:
	\begin{enumerate}[(i)]
		\item binary quantum hypothesis testing (Section~\ref{sec:HT}),
		\item entanglement-assisted classical communication over point-to-point quantum channels (Section~\ref{sec:EA}),
		\item classical data compression with quantum side information (Section~\ref{sec:CQSW}),
		\item entanglement-assisted and unassisted classical communication over quantum multiple-access channels (Section~\ref{sec:MAC}),
		\item entanglement-assisted and unassisted classical communication over quantum broadcast channels (Section~\ref{sec:broadcast}),
		\item entanglement-assisted and unassisted classical communication over quantum  channels with casual state information available at the encoder (Section~\ref{sec:SI}).
	\end{enumerate}
	We refer the reader to the summary given in Table~\ref{table:Table} below.
	
\end{enumerate}

Lastly, the established simple analysis applies to the \emph{position-based coding}, a pivotal technique in one-shot quantum information theory (see, e.g.~Refs.~\cite{ADJ17, AJW19a, AJW19a, Wil17b, QWW18}), proposed by Anshu, Jain, and Warsi \cite[Lemma 4]{AJW19a}, whose decoding strategy again relies on the Hayashi--Nagaoka operator inequality in Eq.~\eqref{eq:HN03}.
The sharpened position-based coding (Theorem~\ref{lemm:packing} in Section~\ref{sec:application}) constitutes the primary technique of deriving numerous one-shot achievability bounds in Section~\ref{sec:application}.
By virtue of its variability, we may term it as a \emph{one-shot quantum packing lemma}, and it might lead to more fruitful applications elsewhere.


\begin{figure}[ht!]
	\centering
	\begin{subfigure}[b]{0.3\textwidth}
		\centering
		\includegraphics[width=\textwidth]{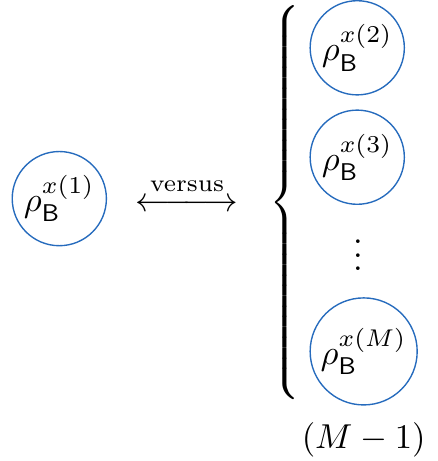}
		\caption{ \small
			Given a realization of a codebook $\mathcal{C}=\{x(1),\ldots, x(M)\}$, the error probability of sending message $m=1$ using pretty-good measurement (PGM) is upper bounded by the error of distinguishing $\rho_{\B}^{x(1)}$ against the remaining states, i.e.~$\sum_{\bar{m}=2}^M \rho_{\B}^{x(\bar{m})}$.			
			We take the sum of the remaining channel output states is because the PGM effectively works as a \emph{one-versus-rest}  classification strategy.
		}
		\label{fig:simple1}
	\end{subfigure}
	\hfill
	\begin{subfigure}[b]{0.3\textwidth}
		\centering
		\includegraphics[width=\textwidth]{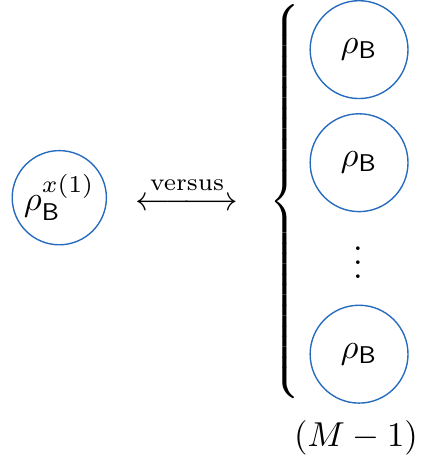}
		\caption{
			Taking the conditional expectation over the random codebook $\mathcal{C}$ conditioned on codeword $x(1)$, the error probability of sending message $m=1$ is upper bounded by the error of distinguishing $\rho_{\B}^{x(1)}$ against $(M-1) \rho_{\B}$.
			Namely, by randomly drawing a codeword $x(1)\sim~p_{\X}$, we are distinguishing the associated channel output state against the average channel output state scaled by a factor $(M-1)$.
		}
		\label{fig:simple2}
	\end{subfigure}
	\hfill
	\begin{subfigure}[b]{0.3\textwidth}
		\centering
		\includegraphics[width=\textwidth]{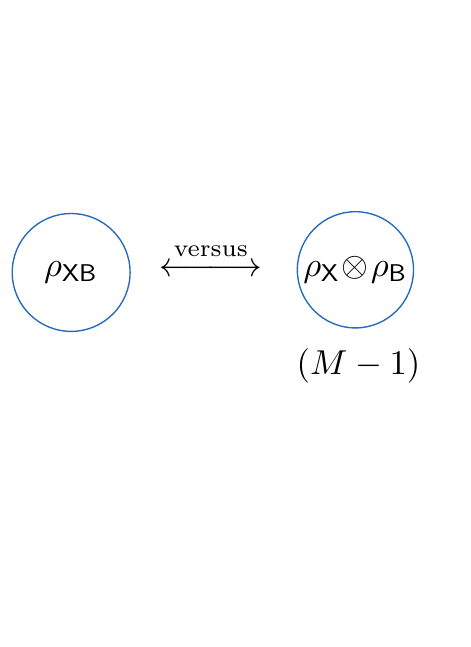}
		\caption{
			Taking the expectation over the transmitted codeword $x(1)\sim~p_{\X}$,
			Figure~\ref{fig:simple2} is equivalent to distinguishing the joint state $\rho_{\X\B}$ between channel input and output against the
			scaled product of its marginal states $(M-1)\rho_{\X}\otimes \rho_{\B}$; 
			see Eq.~\eqref{eq:3}.
			This may be viewed as a \emph{one-shot packing lemma} for classical-quantum channel coding.
		}
		\label{fig:simple3}
	\end{subfigure}
	\caption{\small
		Schematic illustration of the proposed achievability analysis for classical-quantum channel ($x\mapsto \rho_{\B}^x$) coding.
		Similar reasoning applies to various quantum information-theoretic tasks. The reader can refer to Table~\ref{table:Table} and Section~\ref{sec:application} for instances.
	}
	\label{fig:simple}
\end{figure}

\begin{figure}[h!]
	\centering
	\resizebox{1\columnwidth}{!}{
				\includegraphics{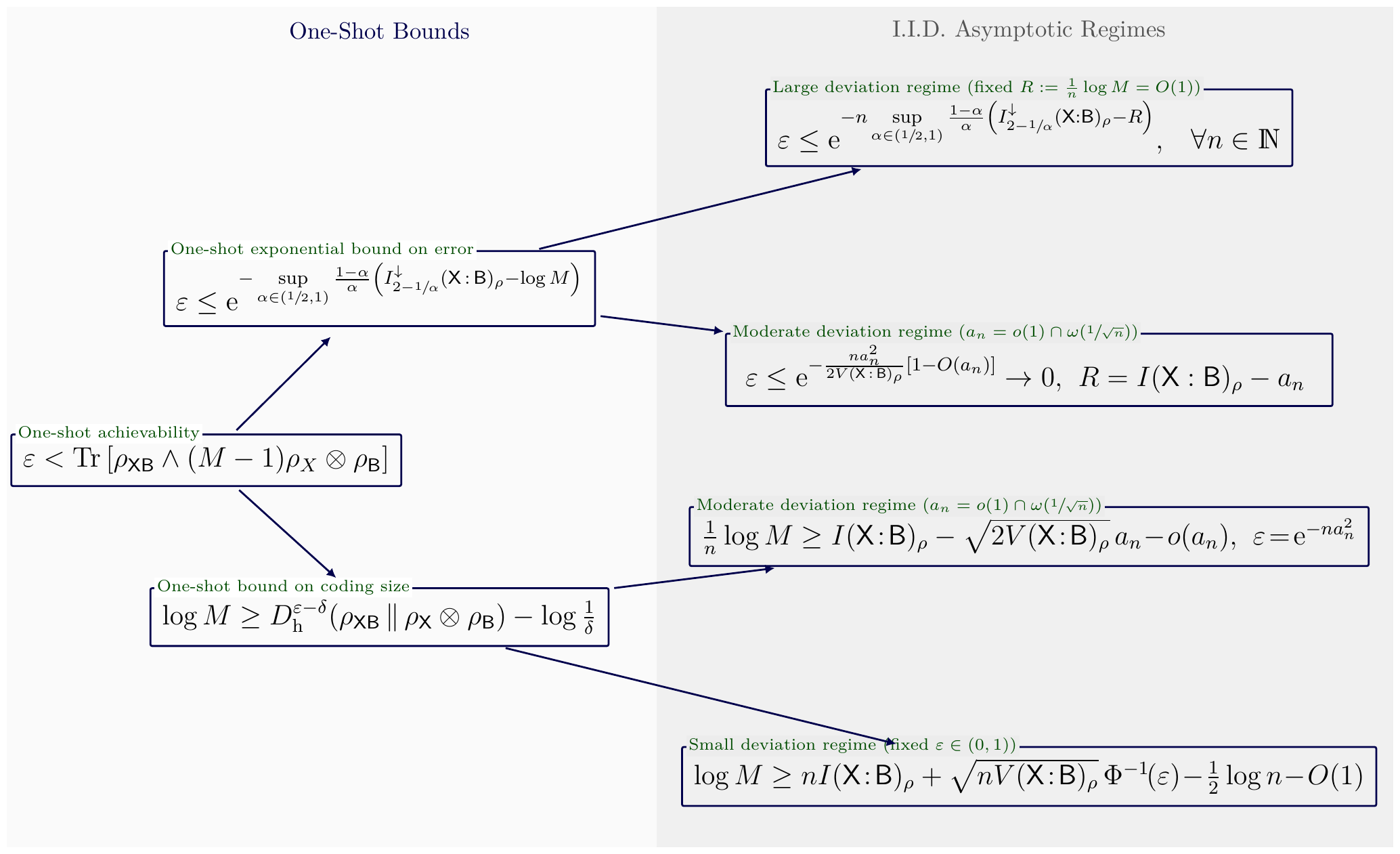}     
	}
	\caption{\small Flow chart of the implications of the established one-shot achievability bound in the large deviation, small deviation, and moderate deviation regimes.
		Here, $R:= \frac1n \log M$ denotes the coding rate for blocklength $n\in\mathds{N}$.
		The precise notation is given in Section~\ref{sec:Main}.
	}
	\label{fig:flow}
\end{figure}

\bigskip
\medskip
\begin{table}[h!]
	\resizebox{\columnwidth}{!}{
		\begin{tabular}{ c  c  c  } 
			\toprule
			
			\multirow{2}{*}{Information-theoretic tasks} & \multirow{2}{*}{\textbf{One-shot achievability}} & Bounds on coding error \\
			
			\cdashlinelr{3-3}
			
			& & Bounds on coding size \\
			
			\midrule
			\midrule
			
			\multirow{2}{*}{Point-to-point quantum channel} & \multirow{2}{*}{$\displaystyle\varepsilon \leq \Tr\left[ \mathscr{N}_{\!\A\to\B}\left(\theta_{\X\A}\right) \wedge (M-1)\theta_{\X}\otimes \mathscr{N}_{\!\A\to\B}\left(\theta_{\A}\right)\right]$} & $\eps \leq \e^{- 
				\sup\limits_{\alpha\in(\sfrac12,1)} 
				\frac{1-\alpha}{\alpha} \left( I_{2-\sfrac{1}{\alpha}}^\downarrow (\X{\,:\,}\B)_{\mathscr{N}_{\!\A\to\B}\left(\theta_{\X\A}\right)} - \log M \right) }$
			\\
			\cdashlinelr{3-3}
			& & $\log M \geq I_\textnormal{h}^{\eps-\delta}\left(\X {\,:\,} \B\right)_{\mathscr{N}_{\!\A\to\B}(\theta_{\X\A})}  - \log \frac{1}{\delta}$\\
			
			\cmidrule(lr){1-3}
			
			\multirow{2}{*}{\shortstack{Entanglement-assisted\\point-to-point quantum channel}} & \multirow{2}{*}{$\displaystyle\varepsilon \leq \Tr\left[ \mathscr{N}_{\!\A\to\B}\left(\theta_{\R\A}\right) \wedge (M-1)\theta_{\R}\otimes \mathscr{N}_{\!\A\to\B}\left(\theta_{\A}\right)\right]$} & $\eps \leq \e^{- 
				\sup\limits_{\alpha\in(\sfrac12,1)} 
				\frac{1-\alpha}{\alpha} \left( I_{2-\sfrac{1}{\alpha}}^\downarrow (\R{\,:\,}\B)_{\mathscr{N}_{\!\A\to\B}\left(\theta_{\R\A}\right)} - \log M \right) }$
			\\
			\cdashlinelr{3-3}
			& & $\log M \geq I_\textnormal{h}^{\eps-\delta}\left(\R {\,:\,} \B\right)_{\mathscr{N}_{\!\A\to\B}\left((\theta_{\R\A}\right)}  - \log \frac{1}{\delta}$\\
			
			\cmidrule(lr){1-3}
			
			\multirow{2}{*}{\shortstack{Quantum channel\\with casual state information}} &
			$\displaystyle\varepsilon \leq \Tr\left[ \mathscr{N}_{\!\A\mathsf{S}\to\B}\left(\theta_{\U\A\mathsf{S}}\right) \wedge (M-1)\theta_{\U}\otimes \mathscr{N}_{\!\A\mathsf{S}\to\B}\left(\theta_{\A\mathsf{S}}\right)\right]$
			& $\eps \leq \e^{- 
				\sup\limits_{\alpha\in(\sfrac12,1)} 
				\frac{1-\alpha}{\alpha} \left( I_{2-\sfrac{1}{\alpha}}^\downarrow (\U{\,:\,}\B)_{\mathscr{N}_{\!\A\mathsf{S}\to\B}\left((\theta_{\U\A\mathsf{S}}\right)} - \log M \right) }$
			\\
			\cdashlinelr{3-3}
			& 
			\multicolumn{1}{l}{$\forall\,\theta_{\U\A\mathsf{S}}:$ $\Tr_{\A}\left[\theta_{\U\A\mathsf{S}}\right] = \theta_{\U} \otimes \vartheta_{\mathsf{S}}$}
			& $\log M \geq I_\textnormal{h}^{\eps-\delta}\left(\U {\,:\,} \B\right)_{\mathscr{N}_{\!\A\mathsf{S}\to\B}(\theta_{\U\A\mathsf{S}})}  - \log \frac{1}{\delta}$\\	
			
			\cmidrule(lr){1-3}
			
			\multirow{2}{*}{\shortstack{Entanglement-assisted quantum channel\\with casual state information}} & 
			$\displaystyle\varepsilon \leq \Tr\left[ \mathscr{N}_{\!\A\mathsf{S}\to\B}\left(\theta_{\R\A\mathsf{S}}\right) \wedge (M-1)\theta_{\R}\otimes \mathscr{N}_{\!\A\mathsf{S}\to\B}\left(\theta_{\A\mathsf{S}}\right)\right]$
			& $\eps \leq \e^{- 
				\sup\limits_{\alpha\in(\sfrac12,1)} 
				\frac{1-\alpha}{\alpha} \left( I_{2-\sfrac{1}{\alpha}}^\downarrow (\R{\,:\,}\B)_{\mathscr{N}_{\!\A\mathsf{S}\to\B}\left((\theta_{\R\A\mathsf{S}}\right)} - \log M \right) }$
			\\
			\cdashlinelr{3-3}
			& 
			\multicolumn{1}{l}{$\forall\,\theta_{\R\A\mathsf{S}}:$ $\Tr_{\A}\left[\theta_{\R\A\mathsf{S}}\right] = \theta_{\R} \otimes \vartheta_{\mathsf{S}}$}
			& $\log M \geq I_\textnormal{h}^{\eps-\delta}\left(\R {\,:\,} \B\right)_{\mathscr{N}_{\!\A\mathsf{S}\to\B}(\theta_{\R\A\mathsf{S}})}  - \log \frac{1}{\delta}$\\	
			
			\cmidrule(lr){1-3}
			
			\multirow{2}{*}{\shortstack{Broadcast quantum channel}} & \multicolumn{2}{l}{$\displaystyle \varepsilon_{\B} \leq \Tr\left[ 
				\Tr_{\mathsf{C}}\left[ \mathscr{N}_{\A\to\B\mathsf{C}}\left( \theta_{\U\A} \right)\right] \wedge (M_{\B}-1) \theta_{\U} \otimes 	\Tr_{\mathsf{C}}\left[ \mathscr{N}_{\A\to\B\mathsf{C}}\left( \theta_{\A} \right)\right]
				\right]$} 
			\\
			& \multicolumn{2}{l}{$\displaystyle \varepsilon_{\mathsf{C}} \leq \Tr\left[ 
				\Tr_{\mathsf{B}}\left[ \mathscr{N}_{\A\to\B\mathsf{C}}\left( \theta_{\V\A} \right)\right] \wedge (M_{\C}-1) \theta_{\V} \otimes 	\Tr_{\mathsf{B}}\left[ \mathscr{N}_{\A\to\B\mathsf{C}}\left( \theta_{\A} \right)\right]
				\right]$
				\,$\forall\, \theta_{\U\V\A}:$ $\Tr_{\A}\left[ \theta_{\U\V\A} \right] = \theta_{\U} \otimes \theta_{\V}$}  \\				
			
			\cmidrule(lr){1-3}
			
			\multirow{2}{*}{\shortstack{Entanglement-assisted\\broadcast quantum channel}} & \multicolumn{2}{l}{$\displaystyle \varepsilon_{\B} \leq \Tr\left[ 
				\Tr_{\mathsf{C}}\left[ \mathscr{N}_{\A\to\B\mathsf{C}}\left( \theta_{\R_{\B}\A} \right)\right] \wedge (M_{\B}-1) \theta_{\R_{\B}} \otimes 	\Tr_{\mathsf{C}}\left[ \mathscr{N}_{\A\to\B\mathsf{C}}\left( \theta_{\A} \right)\right]
				\right]$} 
			\\
			& \multicolumn{2}{l}{$\displaystyle \varepsilon_{\mathsf{C}} \leq \Tr\left[ 
				\Tr_{\mathsf{B}}\left[ \mathscr{N}_{\A\to\B\mathsf{C}}\left( \theta_{\R_{\C}\A} \right)\right] \wedge (M_{\C}-1) \theta_{\R_{\C}} \otimes 	\Tr_{\mathsf{B}}\left[ \mathscr{N}_{\A\to\B\mathsf{C}}\left( \theta_{\A} \right)\right]
				\right]$
				\,$\forall \theta_{\R_{\B}\R_{\C}\A}\,:\,\Tr_{\A}\left[ \theta_{\R_{\B}\R_{\C}\A} \right] = \theta_{\R_{\B}} \otimes \theta_{\R_{\C}}$} \\				
			
			\cmidrule(lr){1-3}
			
			\multirow{2}{*}{\shortstack{Multiple-access quantum channel}} & \multicolumn{2}{l}{$\displaystyle\varepsilon \leq \Tr\left[ \rho_{\X\Y\mathsf{C}} \,\Exterior \left(
				(M_{\A}-1) \rho_{\X}\otimes \rho_{\Y\mathsf{C}}
				+ (M_{\B}-1) \rho_{\Y}\otimes \rho_{\X\mathsf{C}}
				+ (M_{\A}-1)(M_{\B}-1) \rho_{\X}\otimes \rho_{\Y} \otimes \rho_{\mathsf{C}}
				\right)	\right]$} 
			\\
			& \multicolumn{2}{l}{$\displaystyle \rho_{\X\Y\mathsf{C}} :=\mathscr{N}_{\!\A\B\to\mathsf{C}}\left(\theta_{\X\A}\otimes \theta_{\Y\B}\right)$} \\				
			
			\cmidrule(lr){1-3}
			
			\multirow{2}{*}{\shortstack{Entanglement-assisted \\multiple-access quantum channel}} & \multicolumn{2}{l}{$\displaystyle\varepsilon \leq \Tr\left[ \rho_{\R_{\A}\R_{\B}\mathsf{C}} \,\Exterior \left(
				(M_{\A}-1) \rho_{\R_{\A}}\otimes \rho_{\R_{\B}\mathsf{C}}
				+ (M_{\B}-1) \rho_{\R_{\B}}\otimes \rho_{\R_{\A}\mathsf{C}}
				+ (M_{\A}-1)(M_{\B}-1) \rho_{\R_{\A}}\otimes \rho_{\R_{\B}} \otimes \rho_{\mathsf{C}}
				\right)	\right]$} 
			\\
			& \multicolumn{2}{l}{$\displaystyle \rho_{\R_{\A}\R_{\B}\mathsf{C}} :=\mathscr{N}_{\!\A\B\to\mathsf{C}}\left(\theta_{\R_{\A}\A}\otimes \theta_{\R_{\B}\B}\right)$} \\				
			
			\cmidrule(lr){1-3}
			
			\multirow{2}{*}{\shortstack{Classical data compression\\with quantum side information}} & \multirow{2}{*}{$\displaystyle\varepsilon \leq \Tr\left[ \rho_{\X\B} \, \Exterior \left(\frac{1}{M} \mathds{1}_{\X}\otimes \rho_{\B}\right)\right]$} & $\eps \leq \e^{- 
				\sup\limits_{\alpha\in(\sfrac12,1)} 
				\frac{1-\alpha}{\alpha} \left( \log M - H_{2-\sfrac{1}{\alpha}}^\downarrow (\X{\,|\,}\B)_{\rho} \right) }$
			\\
			\cdashlinelr{3-3}
			& & $\log M \leq H_\textnormal{h}^{\eps-\delta}\left(\X {\,|\,} \B\right)_{\rho} + \log \frac{1}{\delta}$\\	
			
			
			\bottomrule
		\end{tabular}
	}
	\caption{\small Summary of the established one-shot achievability bounds in various quantum information-theoretic tasks.
		The precise statements and notation can be found in Section~\ref{sec:application}.
	}	\label{table:Table}	
\end{table}

\newpage
This paper is organized as follows. Section~\ref{sec:minimal} formally introduces the noncommutative minimal and its properties.
Section~\ref{sec:Main} establishes our main result of the one-shot achievability for c-q channel coding; we compare it with existing results in Section~\ref{sec:comparison}.
Section~\ref{sec:application} entails its applications in one-shot quantum information theory. We conclude the paper and discuss possible open problems in Section~\ref{sec:conclusions}.
The Appendix~\ref{sec:proof_trace} proves a useful trace inequality regarding the noncommutative minimal.

\section{The Noncommutative Minimal and Its Properties} \label{sec:minimal}
We first recall the basic concepts of (binary) quantum state discrimination, which constitutes the central tool for the proposed achievability analysis in Section~\ref{sec:Main} below.
Given arbitrary positive semi-definite operators $A$ and $B$, we define the \emph{minimum error} \footnote{
	The quantum state discrimination problems are usually concerned with distinguishing an ensemble of quantum states where each state (i.e.~a density operator with unit trace) in the ensemble is endowed with a prior probability. 
	For example, the \emph{minimum error} defined above coincides with the error probability in conventional quantum state discrimination when assuming $\Tr[A+B] = 1$.
	We note that in Holevo's early works \cite{Hol72, Hol74}, \cite[\S II]{Hol76}, the scenario of distinguishing positive semi-definite operators (even on infinite-dimensional Hilbert space) was studied (see also \cite{AM14} and \cite[\S 3]{Wat18}).
} using two-outcome positive operator-valued measures (POVM) to distinguish them as
\begin{align} \notag
	\inf_{0\leq T \leq \mathds{1}} \Tr\left[A\left(\mathds{1}-T\right)\right] + \Tr\left[ B T\right].
\end{align}
The well-known \emph{Holevo--Helstrom theorem} \cite{Hel67, Hol72, Hol74, Hol76} shows that the infimum can be attained by a Neyman--Pearson test $T = \{A - B >0 \}$ that projects onto the positive part of the difference $A-B$, and the minimization is given by its dual formulation of a semi-definite program (SDP) \cite{YKF05}, \cite[\S 1.2.3]{Wat18}:
\begin{align}
	\begin{split} \label{eq:HH}
		\inf_{0\leq T \leq \mathds{1}} \Tr\left[A\left(\mathds{1}-T\right)\right] + \Tr\left[ B T\right] &= \sup_{M=M^\dagger} \left\{ \Tr\left[M\right] : M\leq A, \, M\leq B  \right\} \\
		&= \Tr\left[ A\wedge B \right].
	\end{split}
\end{align}
Here the supremum is attained by the so-called \emph{noncommutative minimal} (i.e.~the operator with the greatest trace among the lower bounds in terms of the Loewner partial ordering)
\cite{Hel67, Hol72, Hol74, Hol76} of self-adjoint operators $A$ and $B$ \footnote{One can also define the noncommutative minimal among multiple self-adjoint operators. Throughout this paper, we will only consider the case of two positive semi-definite operators.}, i.e.~
\begin{align} \label{eq:defn_min}
	A\wedge B \in \argmax_{M = M^\dagger } \left\{ \Tr\left[M\right] : M\leq A, \, M\leq B  \right\}.
\end{align}
In other words, the {Holevo--Helstrom theorem} \cite{Hel67, Hol72, Hol76} provides an operational meaning to the noncommutative minimal ``$\wedge$" for characterizing the minimum \emph{error} of distinguishing positive semi-definite operators $A$ and $B$  \footnote{In quantum state discrimination, one sometimes studies the maximum success probability \cite{Hel67, Hol72, Hol76} and expresses it via a semi-definite programming formulation as the trace of the so-called \emph{noncommutative maximal} \cite{YKL75, AM14, Wat18}, i.e.~the least (in trace ordering) of the upper bound (in Loewner partial ordering). Since the present paper aims to relate the \emph{error} of a quantum information-theoretic task to that of quantum state discrimination, we will only focus on the error.
	Note also that for distinguishing multiple operators, say $\{A_i \}_i$, the error of the discrimination is given by the noncommutative minimal among the set of operators $\{\sum_{j\neq i} A_j \}_i$ \cite[\S II]{YKL75}.}.
We adopt such an interpretation subsequently.

The main goal of this paper is to characterize the average error probability in quantum information-theoretic tasks in terms of the noncommutative minimal ``$\wedge$". 
To that end, we first review the important properties that will be used in the proposed analysis.
We note that the following properties can be found in existing literature.

\begin{fact}[Properties of noncommutative minimal \footnote{Some of the properties can be proved by using the quantum Hockey--Stick divergence \cite{PV10a, SW12, HRF22} as well. In this paper, we will just stick to the notation of the noncommutative minimal.}] \label{fact}
	Considering arbitrary self-adjoint operators $A$ and $B$,  the noncommutative minimal defined in Eq.~\eqref{eq:defn_min} has the following properties.
	\begin{enumerate}[(i)]
		\item\label{fact:closed-form} (Unique closed-form expression.) 
		The noncommutative minimal $A\wedge B$ is unique and  
		\[
		\displaystyle A\wedge B=  \tfrac{1}{2}\left(A+B - |A-B|\right).
		\]
		
		\vspace{1pt}
		
		\item\label{fact:order} (Monotone increase in the Loewner ordering.) 
		It holds that $\displaystyle\Tr[A\wedge B] \leq \Tr[A'\wedge B']$ for $A\leq A'$ and $B\leq B'$.
		
		\vspace{1pt}
		
		\item\label{fact:monotone} (Monotone increase under positive trace-preserving maps.)
		It holds that $\displaystyle\Tr[A\wedge B] \leq \Tr[\mathscr{N}(A)\wedge \mathscr{N}(B)]$ for any positive trace-preserving map $\mathscr{N}$.
		
		\vspace{1pt}
		
		\item\label{fact:concave} (Concavity.)
		The map $(A,B)\mapsto \Tr[A\wedge B]$ is jointly concave.
		
		\vspace{1pt}
		
		\item\label{fact:direct} (Direct sum.) 
		It holds that $\displaystyle(A\oplus A')\wedge (B\oplus B') = (A\wedge B) \oplus (A'\wedge B')$ for any self-adjoint $A'$ and $B'$.
		
		\vspace{1pt}
		
		\item\label{fact:upper} (Upper bound.)  
		It holds that $\displaystyle\Tr[A\wedge B] \leq \Tr[ A^{1-s} B^s]$ for any $A,B\geq 0$ and $s\in(0,1)$.
		
		\vspace{1pt}		
		\item\label{fact:lower} (Lower bound \footnote{In case that $A+B$ is not invertible, one just uses the Moore--Penrose pseudo-inverse of $A+B$ in the definition of the noncommutative quotient Eq.~\eqref{eq:quotient}.}.)
		It holds that  $\displaystyle\Tr[A\wedge B] \geq \Tr\left[ A \frac{B}{A+B} \right]$ for any $A,B\geq 0$.
	\end{enumerate}
\end{fact}
\begin{proof}
	For \ref{fact:closed-form}: the uniqueness (also for multiple operators) was proved by Holevo \cite[Theorem 2]{Hol74}, \cite[\S 2.2]{Hol76} and later also by Audenaert and Mosonyi \cite[Theorem A.3 \& Eq.~(85)]{AM14}; the closed-form expression may have already been known by Holevo and Helstrom \cite{Hel67, Hol72, Hol76} (see also Ref.~\cite[Lemma A.7]{AM14}).
	
	\noindent Property~\ref{fact:order} follows directly from the definition given in Eq.~\eqref{eq:defn_min} (see also \cite[Lemma A.8]{AM14}).
	
	\noindent Properties~\ref{fact:monotone} and ~\ref{fact:concave} follow from the fact that the trace norm (i.e.~$\|M\|_1:=\Tr[|M|]$) is contractive under positive trace-preserving maps and the triangle inequality (see e.g.~Ref.~\cite[Theorems 9.2 \& 9.3]{NC09}) of $\|\cdot\|_1$.
	We note that the monotone increase under a positive  trace-preserving map and the concavity for multiple operators also hold.
	
	\noindent Property~\ref{fact:direct} with trace is due to the direct-sum structure of the trace norm; the case without trace was proved in Ref.~\cite[Lemma A.9]{AM14}.
	
	\noindent Property~\ref{fact:upper} is the celebrated inequality of Audenaert \textit{et al.} \cite{ACM+07, ANS+08, Aud14} used in proving the \emph{quantum Chernoff bound}; later, it was generalized to infinite-dimensional Hilbert space \cite{JOP+12}.
	
	\noindent Property~\ref{fact:lower} in a special case of $\Tr[A+B] = 1$ is an immediate consequence of the Barnum--Knill Theorem \cite{BK02}, \cite[Theorem 3.10]{Wat18}. That is, the error probability using the {pretty good measurement} \cite{Bel75, HW94} is no larger than twice that of using the optimal measurement. 
	The proof for the general case of $A,B\geq 0$ can be found in the author's previous work \cite[Lemma 3]{SGC22a}.
	For the completeness, we provide an alternative proof of  property~\ref{fact:lower} (and a strengthened result of it) in Appendix~\ref{sec:proof_trace}.
\end{proof}


\section{Main Result: A One-Shot Achievability for Classical-Quantum Channel Coding} \label{sec:Main}

In this section, we prove our main result of establishing a one-shot achievability bound for classical-quantum channel coding via a direct application of the pretty-good measurement (PGM) \cite{Bel75, HW94}.


\begin{defn}[Classical-quantum channel coding]
	Let $\mathscr{N}_{\X \to \B}: x\mapsto \rho_{\B}^x$ be a classical-quantum channel, where each channel output $\rho_{\B}^x$ is a density operator (i.e.~a positive semi-definite operator with unit trace).
	\begin{enumerate}[1.]
		\item Alice holds classical registers $\mathsf{M}$ and $\X$, and Bob holds a quantum register $\B$.
		
		
		\item An encoding $m \mapsto x(m)$ maps equiprobable messages in $\mathsf{M}$ to a codeword in $\X$.
		
		\item The classical-quantum channel $\mathscr{N}_{\X\to \B}$ is applied on Alice's register $\X$ and outputs a state on $\B$ at Bob.
		
		\item A decoding measurement described by a positive operator-valued measure (POVM) $\left\{  {\Pi}_{\B}^{m} \right\}_{m\in\mathsf{M}}$ is performed on Bob's quantum register $\B$ to extract the sent message $m$.
	\end{enumerate}
	
	An $(M, \eps)$-code for $\mathscr{N}_{\X\to \B}$ is a protocol such that $|\mathsf{M}| = M$ and the average error probability satisfies
	\begin{align} \notag
		\frac{1}{M} \sum_{m\in \mathsf{M}}	\Tr\left[ \rho_{\B}^{x(m)} \left(\mathds{1}_{\B} -   {\Pi}_{\B}^{m} \right) \right]
		\leq \varepsilon.
	\end{align}
\end{defn}


The encoding is the standard random coding strategy.
\begin{itemize}
	\item \textbf{Encoding.} 
	Consider a random codebook $\mathcal{C} = \{x(1), x(2), \ldots, x(M)\}$, where each of the codewords $x(m) \in \X$ is pairwise independently drawn from a probability distribution $p_{\X}$.
	Alice sends codewords according to the realization of the codebook $\mathcal{C}$.
	
	\item \textbf{Decoding.} At the receiver, given a realization of the random codebook $\mathcal{C}$ and the corresponding channel output states $\{\rho_{\B}^{x(m)} \}_{ m \in \mathsf{M}} $,
	Bob performs the PGM to decode each message $m\in\mathsf{M}$:
	\begin{align} \label{eq:PGM}
		 {\Pi}^{m}_{\B} := \frac{\rho_{\B}^{x(m)}}{\sum_{\bar{m}\in \mathsf{M}} \rho_{\B}^{x(\bar{m})} }, \quad \forall m\in\mathsf{M}.
	\end{align}
\end{itemize}

Our main result is the following.

\begin{theo}[A one-shot achievability bound for classical-quantum channel coding] \label{theo:main}
	Consider an arbitrary classical-quantum channel $\mathscr{N}_{\X \to \B}: x\mapsto \rho_{\B}^x$. Then, there exists an $(M,\eps)$-code for $\mathscr{N}_{\X \to \B}$ such that for any probability distribution $p_{\X}$,
	\begin{align} \label{eq:main}
		\eps
		\leq \Tr\left[ \rho_{\X\B} \wedge (M-1) \rho_{\X}\otimes \rho_{\B} \right].
	\end{align}
	Here, $\rho_{\X\B} := \sum_{x\in\X} p_{\X} (x) |x\rangle \langle x| \otimes \rho_{\B}^x$ and the noncommutative minimal is $A\wedge B = \frac12(A+B-|A-B|)$ (see Fact~\ref{fact:closed-form}).
\end{theo}

\begin{proof}
	The claim follows from the lower bound of the noncommutative minimal ``$\wedge$'' given in Fact~\ref{fact:lower} for relating the pretty good measurement to the optimal measurement, and the concavity of ``$\wedge$'', i.e.~Fact~\ref{fact:concave}.
	Precisely, 
	given any realization of codebook $\mathcal{C} = \{ x(m)\}_{m\in \mathsf{M}}$, we calculate the average probability of erroneous decoding  using the PGM given in Eq.~\eqref{eq:PGM} as
	\begin{align}
		\frac{1}{M} \sum_{m \in \mathsf{M}} \Tr\left[ \rho_{\B}^{x(m)} \frac{ \sum_{\bar{m}\neq m} \rho_{\B}^{x(\bar{m})} }{\rho_{\B}^x + \sum_{\bar{m}\neq m} \rho_{\B}^{\bar{m}}} \right]
		&{\leq} \frac{1}{M} \sum_{m \in \mathsf{M}}  \Tr\left[  \rho_{\B}^{x(m)} \wedge \left(\sum\nolimits_{\bar{m}\neq m}  \rho_{\B}^{x(\bar{m}) } \right) \right], \label{eq:Ogawa}
	\end{align}	
	where we have applied Fact~\ref{fact:lower} 
	with $A =  \rho_{\B}^{x(m)}$ and $B = \sum\nolimits_{\bar{m}\neq m}  \rho_{\B}^{x(\bar{m})}$
	to relate the error probability under the PGM to the expression in terms of the noncommutative minimal.
	Next, we take the expectation for each $x(m)\sim p_{\mathsf{X}}$ to bound the expected average error probability (which is also called the random-coding error probability), i.e.
	\begin{align}
		&
		\frac{1}{M}\sum_{ m \in \mathsf{M}}  \mathds{E}_{x(m), x(\bar{m}) \sim p_{\X}   }\Tr\left[ \rho_{\B}^{x(m)} \wedge \left(\sum\nolimits_{\bar{m}\neq m}  \rho_{\B}^{x(\bar{m}) } \right) \right] 
		\notag
		\\
		&\overset{\text{(a)}}{\leq} \frac{1}{M}\sum_{ m \in \mathsf{M}} \mathds{E}_{ x(m)\sim p_{\X} }\Tr\left[ \rho_{\B}^{x(m)} \wedge \left( \mathds{E}_{x(\bar{m}) | x(m)} \left[\sum\nolimits_{\bar{m}\neq m}  \rho_{\B}^{x(\bar{m})} \right] \right) \right] \label{eq:main_temp}\\
		\notag
		&\overset{\text{(b)}}{=} \frac{1}{M}\sum_{ m \in \mathsf{M}} \mathds{E}_{x(m)\sim p_{\X}}\Tr\left[  \rho_{\B}^{x(m)} \wedge (M-1)\rho_{\B}\right]
		\\
		\notag
		&= \mathds{E}_{x\sim p_{\X}}\Tr\left[\rho_{\B}^x \wedge (M-1)\rho_{\B}\right],
	\end{align}
	where in (a) we used the concavity given in 
	Fact~\ref{fact:concave}
	and in (b) 
	we recalled the pairwise independence of the random codebook.
	
	Invoking the direct sum formula given in Fact~\ref{fact:direct}, we arrive at the claimed inequality at the right-hand side of Eq.~\eqref{eq:main}.
	Lastly, since the random-coding error probability using any $p_{\X}$ is larger than the error probability of the optimal code, the proof is completed.
\end{proof}

Below, we provide a detailed explanation of how PGM works.
An important feature of PGM is that the POVM element $ {\Pi}_{\B}^{x(m)}$ given in Eq.~\eqref{eq:PGM} is \emph{proportional} to the sent state $\rho_{\B}^{x(m)}$.
On the other hand, the complement of the POVM element, i.e.~$\mathds{1}_{\B} -  {\Pi}_{\B}^{x(m)}$, is proportional to the sum of the remaining states $\sum_{\bar{m}\neq m} \rho_{\B}^{x(\bar{m})}$.
Hence, the average error probability of discriminating $M$ channel output states, i.e.~the left-hand side of Eq.~\eqref{eq:Ogawa}, is equivalent to the error of deciding each sent state $\rho_{\B}^{x(m)}$ using the following two-outcome PGM:
\begin{align} \notag
	\left\{  \frac{ \rho_{\B}^{ x(m)} }{\rho_{\B}^{x(m)} + \sum_{\bar{m}\neq m} \rho_{\B}^{x(\bar{m})}} , \frac{ \sum_{\bar{ m}\neq m} \rho_{\B}^{x(\bar{m})} }{\rho_{\B}^{x(m)} + \sum_{\bar{m}\neq m} \rho_{\B}^{x(\bar{m})} } \right\}.
\end{align}
Such the discrimination between $\rho_{\B}^{x(m)}$ with prior probability $\frac{1}{M}$ against the {sum} of the remaining states (again each with prior probability $\frac{1}{M}$) reflects the nature of the union bound inherited in the PGM; cf.~the right-hand side of Eq.~\eqref{eq:union}. 
Next, taking the expectation on the remaining states ensures that we are discriminating $\rho_{\B}^{x(m)}$ with prior probability $\frac{1}{M}$ against $(M-1)$ identical marginal states $\rho_{\B}$ each with prior probability $\frac{1}{M}$.
Equivalently, this amounts to a binary hypothesis testing between $\rho_{\B}^{x(m)}$ with prior probability $\frac{1}{M}$ against the marginal states $\rho_{\B}$ with prior probability $\frac{M-1}{M}$.
Lastly, after taking the summation over $m\in\mathsf{M}$, 
the above is equal to the discrimination of the joint state $\rho_{\X\B}$ against the scaled decoupled product state $(M-1)\rho_{\X}\otimes \rho_{\B}$.
(See Figure~\ref{fig:simple} for the illustration.)
We hope that this simple proof provides a conceptually clear elucidation on the intimate relation between classical-quantum channel coding and quantum hypothesis testing in a pedagogical way.

\begin{remark}
	In the classical case where $\{\rho_{\B}^x\}_{x\in\X}$ mutually commute, Theorem~\ref{theo:main} reduces to a result by Polyanskiy \cite[Eq.~(2.121)]{Pol10}, which is only $1$-bit weaker than the \emph{dependence testing bound} by Polyanskiy, Poor, and Verd{\'u} \cite[Theorem 17]{PPV10}.
\end{remark}

\begin{remark} \label{remark:strengthened}
	As we show shortly in Sections~\ref{sec:comparison} and \ref{sec:application}, the one-shot bound established in Theorem~\ref{theo:main} already implies (and sharpens) various previously known achievability results in the so-called \emph{achievable rate region}, i.e.~rates below the quantum mutual information.
	Is the bound in Theorem~\ref{theo:main} tight outside the achievable rate region?
	Taking c-q channel coding as an example, when the message size is too large or the coding rate (i.e.~$R:=\frac1n \log M$) is way above the mutual information with respect to $\rho_{\X\B}$, the one-shot bound in Theorem~\ref{theo:main} might not be very tight.
	If $\log (M-1) \geq D_\infty^*(\rho_{\X\B} \,\Vert\, \rho_{\X}\otimes \rho_{\B})$, where $D_\infty^*(A\,\|\,B ):= \inf \left\{\gamma\in \mathds{R} : A\leq \mathrm{e}^\gamma B \right\}$ is the \emph{max-relative entropy} \cite{Dat09, Ren05, KRS09}, then Eq.~\eqref{eq:main} yields a trivial bound: $\eps \leq 1$.
	
	In view of this, Theorem~\ref{theo:main} can be strengthened to the following more involved form:
	\begin{align} \label{eq:tighter}
		\varepsilon
		\leq \left(1 - \frac{1}{M} \Tr\left[ \rho_{\X\B} \wedge (M-1) \rho_{\X}\otimes \rho_{\B} \right] \right) \Tr\left[ \rho_{\X\B} \wedge (M-1) \rho_{\X}\otimes \rho_{\B} \right].
	\end{align}
	Bound \eqref{eq:tighter} follows from the tighter inequality \eqref{eq:tighter_PGM} given in Lemma~\ref{lemm:trace} of Appendix~\ref{sec:proof_trace}, instead of Fact~\ref{fact:lower}.
	Now if $\log (M-1) \geq D_\infty^*(\rho_{\X\B} \,\Vert\, \rho_{\X}\otimes \rho_{\B})$, then the random coding error amounts to randomly guessing equiprobable messages, i.e.~$\eps\leq 1- 1/M$.
	Regardless of the message size $M$, Eq.~\eqref{eq:tighter} is technically a tighter one-shot bound compared to Eq.~\eqref{eq:main}.
	This naturally raises the question of whether Eq.~\eqref{eq:tighter} can lead to a simple proof of the upper bound on the strong converse exponent of c-q channel coding; see \cite[Section 5.4]{MO17}, \cite[Proposition IV.5]{MO18}, and \cite[Proposition VI.2.]{CHDH-2018}. We leave this for future work.
\end{remark}

The established one-shot achievability in Theorem~\ref{theo:main} immediately covers (and sharpens) various known results of deriving the minimal error given a fixed message or coding size $M$ or deriving the maximal message size given a fixed error $\eps$.
Let us define the following two important operational quantities for c-q channel coding:
\begin{align*}
	\eps^\star(\mathscr{N}_{X\to B},M) &:= \inf \left\{ \eps\in \mathds{R} : \exists \text{ an $(M,\varepsilon)$-code for $\mathscr{N}_{X\to B}$}  \right\};
	\\
	M^\star(\mathscr{N}_{X\to B}, \eps) &:= \sup \left\{ M \in \mathds{N} : \exists \text{ an $(M,\varepsilon)$-code for $\mathscr{N}_{X\to B}$} \right\}.
\end{align*}
We note that although $\eps^\star(\mathscr{N},M)$ and $M^\star(\mathscr{N},\eps)$ are inverse functions to each other in the one-shot setting, they lead to different asymptotic expansions in the large deviation and small deviation regimes, respectively.

\begin{prop}[Bounding coding error given fixed coding rate] \label{prop:error}
	Consider an arbitrary classical-quantum channel $\mathscr{N}_{\X \to \B}: x\mapsto \rho_{\B}^x$. 
	Then, for any $n\in\mathbb{N}$ and $R>0$, there exists an $(\mathrm{e}^{nR},\eps)$-code for $\mathscr{N}_{\X \to \B}^{\otimes n}$ such that for any probability distribution $p_{\X}$,
	\begin{align}
		\notag
		\eps \leq \e^{- n \frac{1-\alpha}{\alpha} \left( I_{2-\sfrac{1}{\alpha}}^\downarrow (\X{\,:\,}\B)_{\rho} - R \right) }, \quad \forall \alpha \in (\sfrac12,1).
	\end{align}
	Here, 
	$I_\alpha^\downarrow (\X{\,:\,}\B)_{\rho}  := D_\alpha (\rho_{\X\B}\,\|\,\rho_{\X}\otimes \rho_{\B})$, the state is evaluated on $\rho_{\X\B} := \sum_{x\in\X} p_{\X}(x) |x\rangle\langle x|\otimes \rho_{\B}^x$, and the quantum Petz--R\'enyi divergence \cite{Pet86} is $ D_\alpha (\rho\,\Vert\,\sigma):= \frac{1}{\alpha-1}\log \Tr\left[ \rho^\alpha \sigma^{1-\alpha} \right]$.
	
	The exponent $\sup_{\alpha \in (\sfrac12,1)} \frac{1-\alpha}{\alpha} ( I_{2-\sfrac{1}{\alpha}}^\downarrow (\X{\,:\,}\B)_{\rho} - R )$ is positive if and only if $R > I(\X{\,:\,}\B)_{\rho}:= D (\rho_{\X\B}\,\|\,\rho_{\X}\otimes \rho_{\B})$.
\end{prop}

\begin{proof}
	For the one-shot case $n=1$, we apply Audenaert \emph{et al.}'s inequality, i.e.~Fact~\ref{fact:upper}, on the one-shot bound given in Theorem~\ref{theo:main} with $A = \rho_{\X\B}$, $B = (M-1)\rho_{\X}\otimes
	\rho_{\B}$, and $s = \frac{1-\alpha}{\alpha}$ to obtain the large deviation type bound.
	When considering product channels in the $n$-shot scenario, the exponential decays follows from the fact that $\rho\mapsto I_{2-{1}/{\alpha}}^\downarrow (\X{\,:\,}\B)_{\rho}$ is additive for any $n$-fold product state.
	The positivity holds by noting that the map $\alpha \mapsto I_{2-{1}/{\alpha}}^\downarrow (\X{\,:\,}\B)_{\rho}$ is non-decreasing on $[\sfrac12, 1]$ \cite[Lemma 3.12]{MO17}.
\end{proof}


\begin{prop}[Bounding coding rate given fixed coding error] \label{prop:rate}
	Consider an arbitrary classical-quantum channel $\mathscr{N}_{\X \to \B}: x\mapsto \rho_{\B}^x$. 
	Then, for any $\eps\in(0,1)$ there exists an $\left(M,\eps\right)$-code for $\mathscr{N}_{\X \to \B}$ such that for any probability distribution $p_{\X}$ and any $\delta \in (0,\eps)$,
	\begin{align} \label{eq:M0}
		\log M \geq D_\textnormal{h}^{\eps-\delta}(\rho_{\X\B}\,\|\,\rho_{\X}\otimes \rho_{\B})  - \log \tfrac{1}{\delta}.
	\end{align}
	Here, $D_\textnormal{h}^\varepsilon(\rho \,\|\, \sigma) := \sup_{0\leq T\leq \mathds{1} } \left\{ -\log \Tr[\sigma T] : \Tr[\rho T] \geq 1 - \varepsilon \right\}$ is the $\varepsilon$-\emph{hypothesis testing divergence} \cite{TH13, WR13, Li14}.
	
	Moreover, for any $\eps\in(0,1)$ and for sufficiently large $n\in\mathds{N}$, 
	there exists an $\left(M,\eps\right)$-code for $\mathscr{N}_{\X \to \B}^{\otimes n}$ such that for any probability distribution $p_{\X}$:
	\begin{align}
		\notag		
		\log M \geq n I(\X{\,:\,}\B)_\rho + \sqrt{n V(\X{\,:\,}\B)_\rho } \,   {\Phi}^{-1}(\eps) - \tfrac12\log n - O(1).
	\end{align}
	where $V(\X{\,:\,}\B)_\rho := V(\rho_{\X\B}\,\|\,\rho_{\X}\otimes \rho_{\B})$, $V(\rho\,\|\,\sigma):= \Tr\left[ \rho(\log \rho - \log \sigma)^2 \right] - D(\rho\,\|\,\sigma)^2$, and $  {\Phi}^{-1}(\eps):= \sup\{u: \int_{-\infty}^u \frac{1}{\sqrt{2\pi}} \e^{-\frac12 t^2}\mathrm{d}t \leq \eps \}$ is the inverse of the cumulative distribution of the standard normal distribution.
\end{prop}

\begin{proof}
	By recalling the definition of the noncommutative minimal given in Eq.~\eqref{eq:HH} and by Theorem~\ref{theo:main}, for any test $0\leq T_{\X\B}\leq \mathds{1}_{\X\B}$ satisfying $\Tr[\rho_{\X\B} (\mathds{1}_{\X\B}-T_{\X\B}) ] \leq \eps-\delta$, one has,
	\begin{align}
		\notag
		\eps &\leq \Tr[\rho_{\X\B} (\mathds{1}_{\X\B}-T_{\X\B}) ] + (M-1) \Tr\left[\rho_{\X}\otimes \rho_{\B} T_{\X\B}\right] \\
		\notag
		&\leq \eps - \delta + (M-1) \e^{- D_\textnormal{h}^{\eps - \delta}(\rho_{\X\B}\,\|\,\rho_{\X}\otimes \rho_{\B}) },
	\end{align}
	completing the proof.
	
	The second-order achievability then follows from the expansion of the quantum hypothesis-testing divergence \cite{TH13, Li14, DPR16, OMW19, PW20} by choosing $\delta = \sfrac{1}{\sqrt{n}}$:
	$
	D_\textnormal{h}^{\varepsilon \pm \delta} \left(\rho^{\otimes n} \,\|\, \sigma^{\otimes n} \right) \geq n D\left(\rho \,\|\, \sigma\right) + \sqrt{n V\left(\rho \,\|\, \sigma\right)}   {\Phi}^{-1} (\varepsilon) - O(1).
	$
\end{proof}

We remark that both Propositions~\ref{prop:error} and \ref{prop:rate} extends to the moderate deviation regime by directly following the approaches from \cite{CH17, CTT2017}.
The reader may refer to Figure~\ref{fig:flow} for the corresponding expressions.

\begin{remark}
	Given that Theorem~\ref{theo:main} already provides a one-shot bound on the average error probability, one may wonder why to weaken Eq.~\eqref{eq:main} in Theorem~\ref{theo:main} to obtain another one-shot bound in Proposition~\ref{prop:error} (note that they both have closed-form expressions).
	The reason is that the minimum error in terms of the noncommutative minimal on the right-hand side of Eq.~\eqref{eq:main} is not multiplicative under product states.
	Nevertheless, it can be further upper bounded by certain multiplicative R\'enyi-type quantities. That is exactly the spirit of the quantum Chernoff bound \cite{ACM+07, ANS+08, Aud14, JOP+12}, and hence, we term the result of Proposition~\ref{prop:error} as a kind of \emph{large deviation type bound} \footnote{By large deviations, we meant the scenario where the coding rate $R=\frac1n \log M$ deviates from the fundamental threshold (i.e.~the quantum mutual information or the Holevo quantity when optimizing the input distributions) by a \emph{large amount}; i.e.~namely, it is constant $O(1)$ away from the fundamental limit. 
		In the small deviation regime, the optimal rate $R$ for some fixed error $\varepsilon \in (0,1)$ converges to the fundamental limit at a speed of $O(\sfrac{1}{\sqrt{n}})$, meaning that $R$	deviates by a \emph{small amount}.
		In between, we refer to the \emph{moderate deviation regime}, where $R$ deviates by the order $\omega(\sfrac{1}{\sqrt{n}}) \cap o(1)$ \cite{CH17,CTT2017}. In this case, the optimal error vanishes at a sub-exponential speed.} accordingly.
	
	On the other hand, Theorem~\ref{theo:main} also gives a one-shot and asymptotic expansions in the small deviation regime (Proposition~\ref{prop:rate}).
	Hence, to some extent, Theorem~\ref{theo:main} may be viewed as a ``meta'' achievability for classical communication over quantum channels (see also Theorem~\ref{lemm:packing} later in Section~\ref{sec:EA}).
\end{remark}

\begin{remark} \label{remark:time}
	Most existing one-shot achievability bounds to date (e.g.~\cite{HN03, WR13, Wil17b, AJW19a}) are expressed in terms of the quantum hypothesis testing divergence $D_\text{h}^\eps$ as in Eq.~\eqref{eq:M0} of Proposition~\ref{prop:rate}, for the reason that they directly provide a one-shot characterization (lower bound) on the maximal message or coding size $M$ given a fixed coding error $\varepsilon$, which is also called the \emph{$\varepsilon$-one-shot channel capacity}.
	To numerically compute $D_\text{h}^\eps$, one can formulate the quantity into a standard form of a semi-definite program (SDP); namely, it is an optimization over a $d_{\B}\times d_{\B}$ matrix-valued variable with $m:=d_{\B}^2+1$ linear (scalar) constraints, where we use $d_{\B}$ to denote the dimension of the underlying Hilbert space representing the quantum register $\B$.
	(Here, we only consider the computation on the quantum part  of register $\B$ for simplicity without involving computation on the classical register $\X$.)
	Using the state-of-the-art (classical) SDP solver \cite{HJS+21}, the running time \footnote{
		We use notation $O^*$ to hide $m^{o(1)}$ and $\log\frac{1}{\epsilon}$ factors, where $\epsilon$ is accuracy parameter. 
	} is $O^*(m^{\omega}) = O^*(d_{\B}^{2\omega}) = O^*(d_{\B}^{4.746})$, where $\omega\leq 2.373$ is the exponent of matrix multiplication \cite{DDH+07}.
	
	On the other hand, the one-shot bound provided in Theorem~\ref{theo:main} admits a closed-form expression in terms of the trace norm.
	Using the state-of-the-art algorithm for approximating singular values \cite{BGK19}, it requires running time $O^*(d_{\B}^{\omega} \log^2 d_{\B}) = O^*(d_{\B}^{2.373}\log^2 d_{\B})$.
	This then shows that the computation of the proposed one-shot achievability bound in terms of the noncommutative minimal in Theorem~\ref{theo:main} is nearly \emph{quadratically efficient} compared to the computation of the one-shot bounds in terms of the quantum hypothesis-testing divergence.
\end{remark}

\subsection{Comparison to Existing Results} \label{sec:comparison}

In the following, we compare the implications of the established one-shot achievability bounds, i.e.~Propositions~\ref{prop:error} and \ref{prop:rate} with existing results. 
We refer the reader can refer to Table~\ref{table:comparison} below for a summary.

The exponential decaying rate of error probability given in Proposition~\ref{prop:error} matches the one proved by Hayashi \cite[Eq.~(9)]{Hay07}.
However, in the one-shot setting, the large deviation type bound in Proposition~\ref{prop:error} is tighter than \cite[Eq.~(9)]{Hay07} (without the factor $4$).
If further the $\mathscr{N}_{\X\to \B}$ is a pure-state channel, one has $I_{2-\sfrac{1}{\alpha}}^\downarrow (\X{\,:\,}\B)_{\rho} = I_{{\alpha}} (\X{\,:\,}\B)_{\rho} := \inf_{\sigma_B\in\mathcal{S}(\mathcal{H}_{\B})} D_\alpha(\rho_{\X\B}\,\|\,\rho_{\X}\otimes \rho_{\B})
= \frac{\alpha}{\alpha-1} \log \Tr[ (\sum_{x\in\X} p_{\X}(x) (\rho_{\B}^x)^{\alpha})^{1/\alpha} ]$ (where the minimization is over all density operators on Hilbert space $\mathcal{H}_{\B}$, i.e.$~\mathcal{S}(\mathcal{H}_{\B})$).
Hence, the bound in Proposition~\ref{prop:error} is tighter than the bound proved by Burnashev and Holevo \cite[Proposition 1]{BH98} (without the factor $2$).

Hayashi--Nagaoka \cite[Lemma 3]{HN03}, and Wang--Renner \cite[Theorem 1]{WR13} employed the Hayashi--Nagaoka inequality Eq.~\eqref{eq:HN03} to obtain a one-shot achievability bound \footnote{More precisely, Hayashi and Nagaoka obtained the first one-shot achievability bound (for general c-q channels) as in Eq.~\eqref{eq:M_WR} but in terms of the \emph{information-spectrum divergence} $D_\textnormal{s}^\eps$ \cite{VH94, Han03, HN03, NH07, TH13} instead of the hypothesis-testing divergence $D_\textnormal{h}^\eps$ defined in Proposition~\ref{prop:rate}.
	On the other hand, it is known that $D_\textnormal{h}^\eps(\,\cdot \| \cdot \,) \geq D_\textnormal{s}^\eps(\,\cdot \| \cdot \,)$ \cite[Lemma 12]{TH13}, and hence, the one-shot achievability bound in terms of $D_\textnormal{h}^\eps$ is tighter than that in terms of $D_\textnormal{s}^\eps$.
	Here, we remark that the approach proposed by Hayashi and Nagaoka \cite{HN03} allows for choosing any measurement along with applying the Hayashi--Nagaoka inequality Eq.~\eqref{eq:HN03} in establishing the achievability. Namely, the analysis in \cite{HN03} with Eq.~\eqref{eq:HN03} can already lead to Eq.~\eqref{eq:M_WR}.
	We remark that the terminology and concept of the hypothesis-testing divergence $D_\textnormal{h}^\eps$ might already appear in the contexts of statistical hypothesis testing by Stein--Chernoff \cite{Che56}, Strassen \cite{Str62}, Csisz{\'a}r--Longo \cite{CL71}, 
	Polyanskiy--Poor--Verd{\'u} \cite{PPV10}, and by Wang--Renner \cite{WR13} in the quantum setting.
} on the message or coding size $M$: for any $0<\delta<\eps<1$ with choosing $c = \frac{\delta}{2\eps - \delta}$ in Eq.~\eqref{eq:HN03},
\begin{align} \label{eq:M_WR}
	\log M \geq D_\textnormal{h}^{\eps-\delta}(\rho_{\X\B}\,\|\,\rho_{\X}\otimes \rho_{\B})  - \log \frac{4}{\delta^2}
\end{align}
in terms of the hypothesis-testing divergence $D_\textnormal{h}^{\eps}$ introduced in Proposition~\ref{prop:rate}.
The term $- \log \frac{4}{\delta^2}$ results from optimizing  coefficient $c$ when applying the Hayashi--Nagaoka inequality.
Compared to Eq.~\eqref{eq:M_WR}, the established Proposition~\ref{prop:rate} in Section~\ref{sec:Main} does not need to choose the appropriate coefficient $c$, and hence, it gives a tighter one-shot achievability bound on $M$ (especially when $\delta$ is small):
\begin{align} \label{eq:M}
	\log M \geq D_\textnormal{h}^{\eps-\delta}(\rho_{\X\B}\,\|\,\rho_{\X}\otimes \rho_{\B})  - \log \frac{1}{\delta}.
\end{align}
Specialized to the i.i.d.~asymptotic scenario of $n$-fold product channels with $\delta = \sfrac{1}{\sqrt{n}}$, Eq.~\eqref{eq:M} yields an improved third-order coding rate by a factor $\frac12 \log n$ compared to the asymptotics based on Eq.~\eqref{eq:M_WR}.

Beigi and Gohari \cite{BG14} generalized a superb classical achievability approach by Yassaee \textit{et al.} \cite{YRG13} to establish a one-shot achievability bound on $M$
\cite[Corollary 1]{BG14} as well:
\begin{align} \label{eq:BG}
	\log M \geq D_\textnormal{s}^{\eps-\delta}(\rho_{\X\B}\,\|\,\rho_{\X}\otimes \rho_{\B})  - \log \frac{1-\eps}{\delta},
\end{align}
with $D_\textnormal{s}^\eps(\rho \,\|\, \sigma ):= \sup\left\{ \gamma \in \mathds{R} : \Tr\left[ \rho \left\{ \rho\leq \mathrm{e}^\gamma \sigma \right\} \leq \eps \right]  \right\} $ the information-spectrum divergence \cite{VH94, Han03, HN03, NH07, TH13}.
Comparing Eq.~\eqref{eq:M} to Eq.~\eqref{eq:BG}, we recall the relation between quantum hypothesis-testing divergence $D_\textnormal{h}^\eps$ and quantum information-spectrum divergence $D_\textnormal{s}^\eps$ \cite[Lemma 12]{TH13}:
\begin{align} \label{eq:TH13}
	D_\textnormal{h}^\eps(\rho \,\|\, \sigma )
	\geq 	D_\textnormal{s}^\eps(\rho \,\|\, \sigma )
	\geq D_\textnormal{h}^{\eps-\delta'}(\rho \,\|\, \sigma ) - \log\frac{1}{\delta'},
	\quad \forall \, 0<\delta'< \eps.
\end{align}
This indicates that the proposed one-shot bound Eq.~\eqref{eq:M} has a stronger leading term  $D_\textnormal{h}^{\eps-\delta}$ instead of $D_\textnormal{s}^{\eps-\delta}$ \footnote{
	In the i.i.d.~asymptotic setting, the leading term (i.e.~either $D_\textnormal{h}^{\eps-\delta}$ or $D_\textnormal{s}^{\eps-\delta}$) dominates both the first-order and the second-order coding rates.
	On the other hand, Eq.~\eqref{eq:BG} does have a better constant $-\log (1-\eps)$, which corresponds to the \emph{fourth-order term} in the small deviation regime.
	Such a term is negligible for small errors (say e.g.~$\eps\leq 10^{-3}$). For large errors, one can invoke the strengthened one-shot bound in Eq.~\eqref{eq:tighter} (though the formula is more involved).
}.

When considering the asymptotic expansion of the coding rate in the i.i.d.~setting, one has to translate $	D_\textnormal{s}^\eps$ in Eq.~\eqref{eq:BG} back to 	$D_\textnormal{h}^\eps$ using Eq.~\eqref{eq:TH13} and to employ the second-order achievability \footnote{
	The second-order expansion of the quantum hypothesis-testing divergence was concurrently proposed by Tomamichel--Hayashi \cite{TH13} and Ke Li \cite{Li14}.
	Here in Eq.~\eqref{eq:Li}, we cite Li's result \cite[Theorem 5]{Li14} since it has a better third-order term $- O(1)$ than that of Tomamichel--Hayashi \cite[Eqs.~(28) and (33)]{TH13}, in which the third-order term is $-(\frac12 + \min(\lambda(\sigma), \nu(\sigma )) \log n - O(1)$ for $\lambda(\sigma )$ being the \emph{logarithmic condition number} of $\sigma$ and $\nu(\sigma )$ being the number of distinct eigenvalues of $\sigma$.
	
	We also remark that Li's result Eq.~\eqref{eq:Li} was later generalized to infinite-dimensional separable Hilbert space by Datta--Pautrat--Rouz{\'{e}} \cite[Proposition 1]{DPR16} and Oskouei--Mancini--Wilde \cite[Lemma A.1]{OMW19}.
	The same result for general von Neumann algebras was proved by Pautrat and Wang \cite[Theorem 1]{PW20}.
} 
of the quantum hypothesis-testing divergence $D_\textnormal{h}^\eps$ \cite{TH13, Li14, DPR16, PW20}:
\begin{align} \label{eq:Li}
	D_\textnormal{h}^{\varepsilon \pm \delta} \left(\rho^{\otimes n} \,\|\, \sigma^{\otimes n} \right) \geq n D\left(\rho \,\|\, \sigma\right) + \sqrt{n V\left(\rho \,\|\, \sigma\right)}   {\Phi}^{-1} (\varepsilon) - O(1).
\end{align}
Then, Beigi and Gohari's result, Eq.~\eqref{eq:BG}, leads to
\begin{align} \notag 
	\log M \geq n I(\X{\,:\,}\B)_\rho + \sqrt{n V(\X{\,:\,}\B)_\rho } \,   {\Phi}^{-1}(\eps) -\log n - O(1).
\end{align}
This achieves the same third-order term 
as the asymptotic expansion using Eq.~\eqref{eq:M_WR} and Eq.~\eqref{eq:Li}.
On the other hand, the established Eq.~\eqref{eq:M} with Eq.~\eqref{eq:Li} gives a tighter third-order term of coding rate:
\begin{align*} 
	\log M \geq n I(\X{\,:\,}\B)_\rho + \sqrt{n V(\X{\,:\,}\B)_\rho } \,   {\Phi}^{-1}(\eps) -\frac12\log n - O(1).
\end{align*}

Inspired from the third-order asymptotics of the classical hypothesis-testing divergence proved by Strassen \cite[Theorem 3.1]{Str62} (see also \cite[Lemma 46]{PPV10} and \cite[Proposition 2.3]{Tan14}), we conjecture the following third-order achievability of the quantum hypothesis-testing divergence:
\begin{align} \label{eq:third_HT}
	D_\textnormal{h}^{\varepsilon \pm \delta} \left(\rho^{\otimes n} \,\|\, \sigma^{\otimes n} \right) \geq n D\left(\rho \,\|\, \sigma\right) + \sqrt{n V\left(\rho \,\|\, \sigma\right)}   {\Phi}^{-1} (\varepsilon) + \frac12 \log n - O(1).
\end{align}
If Eq.~\eqref{eq:third_HT} was true, then the established Eq.~\eqref{eq:M} will imply
\begin{align} \label{eq:third_cq}
	\log M \geq n I(\X{\,:\,}\B)_\rho + \sqrt{n V(\X{\,:\,}\B)_\rho } \,   {\Phi}^{-1}(\eps) - O(1).
\end{align}
We remark that Eq.~\eqref{eq:third_cq} will give the best possible achievable third-order coding rate for c-q channel coding without further assumptions on the channel \footnote{
	We note that Altu{\u{g}} and Wagner proved that the third-order coding rate $O(1)$ is optimal for classical symmetric and singular channels \cite[Proposition 1]{AW14_3} (see also \cite[Theorem 4.3]{Tan14}).
	In other words, if the non-singularity condition is not imposed, larger third-order coding rate than $O(1)$ is not possible.
}.

\begin{remark}
	At the writing of this paper, a very recent work by Renes establishes the optimal error exponent for \emph{symmetric} classical-quantum channels \cite{Ren22}.
	The result is asymptotic, but it matches the quantum sphere-packing bound \cite{Dal13, DW14, CHT17} for the high achievable rate region, and hence, it is asymptotically optimal and tighter than the error exponent obtained in Proposition~\ref{prop:error} in the i.i.d.~asymptotic setting for symmetric classical-quantum channels.
	A one-shot bound for that is still missing.
	We leave this for future work.
\end{remark}

\begin{table}[t!]
	\centering
	\resizebox{1\columnwidth}{!}{
		\begin{tabular}{lcc}
			\toprule
			\addlinespace
			& \multicolumn{2}{c}{\textbf{One-shot exponential bounds on coding error}} \\ 
			
			\midrule
			
			\cellcolor{gray!20} Proposition~\ref{prop:error} & 
			\multicolumn{2}{c}{\cellcolor{gray!20}$\displaystyle\eps \leq \e^{- \sup\limits_{\alpha \in (\sfrac12,1)} \frac{1-\alpha}{\alpha} \left( I_{2-\sfrac{1}{\alpha}}^\downarrow (\X{\,:\,}\B)_{\rho} - R \right) }$} \\
			
			\cmidrule{1-1}\cmidrule(lr){2-3}
			Hayashi \cite{Hay07} & \multicolumn{2}{c}{$\displaystyle \eps \leq 4\e^{- \sup\limits_{\alpha \in (\sfrac12,1)} \frac{1-\alpha}{\alpha} \left( I_{2-\sfrac{1}{\alpha}}^\downarrow (\X{\,:\,}\B)_{\rho} - R \right) }$} \\
			
			\cmidrule{1-1}\cmidrule(lr){2-3}
			
			Burnashev--Holevo \cite{BH98} & \multicolumn{2}{c}{\multirow{2}{*}{$\displaystyle \eps \leq 2\e^{- \sup\limits_{\alpha \in (\sfrac12,1)} \frac{1-\alpha}{\alpha} \left( I_{2-\sfrac{1}{\alpha}}^\downarrow (\X{\,:\,}\B)_{\rho} - R \right) }$}} \\
			(pure-state channels) & & \\
			
			\midrule
			\midrule
			
			& \multicolumn{2}{c}{\bf Achievability bounds on coding size} \\ 
			\cmidrule(lr){2-2}\cmidrule(lr){3-3}
			& one-shot bounds & i.i.d.~asymptotic expansion \\
			
			\midrule
			
			\cellcolor{gray!20} Proposition~\ref{prop:rate} & \cellcolor{gray!20} $\displaystyle\log M \geq D_\textnormal{h}^{\eps-\delta}(\rho_{\X\B}\,\|\,\rho_{\X}\otimes \rho_{\B})  - \log \frac{1}{\delta}$ & \cellcolor{gray!20} $\displaystyle\log M \geq n I + \sqrt{n V } \,   {\Phi}^{-1}(\eps) - \tfrac12\log n - O(1)$ \\
			
			\cmidrule{1-1}\cmidrule(lr){2-2}\cmidrule(lr){3-3}
			
			Hayashi--Nagaoka \cite{HN03} & \multirow{2}{*}{$\displaystyle\log M \geq D_\textnormal{h}^{\eps-\delta}(\rho_{\X\B}\,\|\,\rho_{\X}\otimes \rho_{\B})  - \log \frac{4}{\delta^2}$} &  \multirow{2}{*}{$\displaystyle\log M \geq n I + \sqrt{n V } \,   {\Phi}^{-1}(\eps) - \log n - O(1)$} \\
			
			Wang--Renner \cite{WR13} & & \\
			
			\cmidrule{1-1}\cmidrule(lr){2-2}\cmidrule(lr){3-3}
			
			\multirow{2}{*}{Beigi--Gohari \cite{BG14}} & $\displaystyle\log M \geq D_\textnormal{s}^{\eps-\delta}(\rho_{\X\B}\,\|\,\rho_{\X}\otimes \rho_{\B})  - \log \frac{1-\eps}{\delta}$
			&
			\multirow{2}{*}{ $\displaystyle\log M \geq n I+ \sqrt{n V } \,   {\Phi}^{-1}(\eps) - \log n - O(1)$} \\
			
			& $\displaystyle{\color{White}{\log M }} \,\,\geq D_\textnormal{h}^{\eps-2\delta}(\rho_{\X\B}\,\|\,\rho_{\X}\otimes \rho_{\B})  - \log \frac{1-\eps}{\delta^2}$ & \\
			
			\cmidrule{1-1}\cmidrule(lr){2-2}\cmidrule(lr){3-3}
			
			\multirow{2}{*}{Ogawa \cite{Oga15}} & $\displaystyle\log M \geq D_\textnormal{s}^{\eps-\delta}(\rho_{\X\B}\,\|\,\rho_{\X}\otimes \rho_{\B})  - \log \frac{1}{\delta}$
			&
			\multirow{2}{*}{ $\displaystyle\log M \geq n I+ \sqrt{n V } \,   {\Phi}^{-1}(\eps) - \log n - O(1)$} \\
			
			& $\displaystyle{\color{White}{\log M }} \,\,\,\geq D_\textnormal{h}^{\eps-2\delta}(\rho_{\X\B}\,\|\,\rho_{\X}\otimes \rho_{\B})  - \log \frac{1}{\delta^2}$ & \\
			
			\bottomrule
		\end{tabular}
	}
	\caption{Comparisons of the  one-shot achievability bounds on the error probability and the coding size and rate established in Propositions~\ref{prop:error} and \ref{prop:rate} of Section~\ref{sec:Main} with existing results.
		We also present the i.i.d.~asymptotic expansion of the coding rate to highlight the resulting third-order terms, where we shorthand $I\equiv I(\X{\,:\,}\B)_\rho$ and $V\equiv V(\X{\,:\,}\B)_\rho$ for brevity.
	} \label{table:comparison}
\end{table}

\section{Applications in Quantum Information Theory} \label{sec:application}
The analysis proposed in Section~\ref{sec:Main} naturally extends to classical communication over quantum channels, network information theory \cite{GK11}, and beyond; see Table~\ref{table:Table} in Section~\ref{sec:introduction}.
We apply our analysis using the pretty-good measurement to binary quantum hypothesis testing in Section~\ref{sec:HT}.
We present entanglement-assisted classical communication over quantum channels in
Section~\ref{sec:EA}.
Section~\ref{sec:CQSW} is for classical data compression with quantum side information.
Section~\ref{sec:MAC} studies entanglement-assisted and unassisted classical communication over quantum multiple-access channels.
Section~\ref{sec:broadcast} considers entanglement-assisted and unassisted classical communication over quantum broadcast channels.
Section~\ref{sec:SI} is devoted to entanglement-assisted classical communication over quantum channels with casual state information available at the encoder.

\subsection{Binary Quantum Hypothesis Testing} \label{sec:HT}
Binary quantum hypothesis testing and the optimal quantum measurement is a relatively well-studied topic in quantum information theory due to its simpler mathematical structure and operational significance \cite{Hel67, Hol72, Hol76, YKL75, HP91, ON00, OH04, Hay02, ACM+07, ANS+08, Hay07, NS09, JOP+12, TH13, Li14, MO14, HT14, CTT2017, CH17, CHT19, CGH19, Hao-Chung, CWY20}. 
The goal of this section is not to re-do the analysis via optimal measurements, but to show how the sub-optimal pretty-good measurement along with the properties of the noncommutative minimal given in Section~\ref{sec:minimal} can recover the existing results with only a slightly sub-optimal coefficient. Specifically, we will show that pretty-good measurement can also achieve the \emph{quantum Hoeffding bound} \cite[\S 5.5]{ANS+08}. This indicates that the proposed analysis should not be too loose in terms of the one-shot exponential bounds (at least for binary quantum hypothesis testing).

\textit{Symmetric scenario}. We first consider the symmetric scenario, where the two quantum hypotheses are described by density operators $\rho$ and $\sigma$ with prior probability $p\in(0,1)$ and $1-p$, respectively.
Note that the one-shot quantum hypothesis testing is also known as the \emph{quantum state discrimination}; the relation between the optimal measurement (i.e.~the quantum Neyman--Pearson test) and the pretty-good measurement was proved by Barnum and Knill \cite{BK02}, \cite[Theorem 3.10]{Wat18}. 

Subsequently, we show that the lower bound on the noncommutative minimal (Fact~\ref{fact:lower}) can be interpreted as an adaptation of the Barnum--Knill Theorem.
On one hand, the Holevo--Helstrom theorem \cite{Hel67, Hol72, Hol76} shows that the minimal error for distinguishing $\rho$ and $\sigma$ in the symmetric scenario is given by $\Tr\left[ p \rho \wedge (1-p) \sigma \right]$.
On the other hand, by using pretty-good measurement with respect to the weighted states $(p\rho,(1-p)\sigma)$ and applying the lower bound of the noncommutative minimal, Fact~\ref{fact:lower}, the corresponding error probability is given by
\begin{align} \label{eq:HT0}
	p\Tr\left[ \rho \frac{(1-p)\sigma}{p\rho + (1-p)\sigma } \right]
	+ (1-p)\Tr\left[ \sigma \frac{p\rho}{p\rho + (1-p)\sigma} \right]
	\leq 2\Tr\left[ p\rho \wedge (1-p) \sigma \right],
\end{align}
which is twice the error probability compared to the minimal error via the optimal measurement. This coincides with the claim made by Barnum and Knill on the relation between the error probability using the optimal measurement and that of using the pretty-good measurement.

\begin{remark}
	As mentioned in Remark~\ref{remark:strengthened} of Section~\ref{sec:Main}, the upper bound in Eq.~\eqref{eq:HT0} can be strengthened to:
	\begin{align} \notag
		2\left(1 - \Tr\left[ p\rho \wedge (1-q) \sigma \right] \right) \Tr\left[ p\rho \wedge (1-q) \sigma \right],
	\end{align}
	by using Eq.~\eqref{eq:tighter_PGM} of Lemma~\ref{lemm:trace} in Appendix~\ref{sec:proof_trace} instead of Fact~\ref{fact:lower}.
	
	On the other hand, one can also use the pretty-good measurement of the form $\left\{ \frac{\rho}{\rho+\sigma}, \frac{\sigma}{\rho+\sigma}  \right\}$ to obtain an achievable error probability $\left(1 - \Tr\left[\rho \wedge \sigma \right] \right)\Tr\left[\rho \wedge \sigma \right]$ or simply $\Tr\left[\rho \wedge \sigma \right]$.
\end{remark}

\textit{Asymmetric scenario}.
We move on to consider the asymmetric scenario, namely, the tradeoff between the type-I error and the type-II error without knowing the prior distribution.
We use pretty-good measurement $\left\{ \frac{\rho}{\rho+\mu\sigma}, \frac{\mu\sigma}{\rho+\mu\sigma}  \right\}$ with a parameter $\mu$ that will be specified later and apply Fact~\ref{fact:lower} to bound the type-I error $\alpha$ and the type-II error $\beta$:
\begin{align} \label{eq:HT}
	\begin{dcases}
		\alpha = \Tr\left[ \rho \frac{\mu\sigma}{\rho+\mu\sigma} \right]
		\leq \Tr\left[ \rho \wedge \mu \sigma \right], & \\
		\beta = \Tr\left[ \sigma \frac{\rho}{\rho+\mu\sigma} \right]
		\leq \mu^{-1} \Tr\left[ \rho \wedge \mu \sigma \right]. &\\
	\end{dcases}
\end{align}
Next, we will show how Eq.~\eqref{eq:HT} implies both the small deviation type bound and the large deviation type bound.
As in the proof of Proposition~\ref{prop:rate}, we invoke the definition of ``$\wedge$" in Eq.~\eqref{eq:HH} with any test $T$ satisfying $\Tr\left[\rho(\mathds{1}-T)\right]\leq\eps-\delta$, i.e.~
\begin{align}
	\notag
	\Tr\left[ \rho \wedge \mu \sigma \right] &= \inf_{0\leq T\leq \mathds{1}} \Tr\left[\rho(\mathds{1}-T)\right] + \Tr\left[ \mu \sigma T \right]\\
	&\leq \eps-\delta +\mu \mathrm{e}^{- D_\text{h}^{\eps-\delta}(\rho\Vert\sigma)}. \label{eq:HT1}
\end{align}
Choosing $\mu$ such that the right-hand side of Eq.~\eqref{eq:HT1} equals $\eps$ and recalling the upper bound on ``$\wedge$" (Fact~\ref{fact:upper} with taking $s\to0$) to obtain the following bound on the type-II error:
\begin{align} \label{eq:HT2}
	\beta \leq \mu^{-1} \Tr\left[ \rho \wedge \mu \sigma \right] 
	\leq \mu^{-1} = \mathrm{e}^{- D_\text{h}^{\eps-\delta}(\rho\,\Vert\,\sigma) - \log \delta} , \quad \forall\, 0<\delta < \eps< 1.
\end{align}
We remark that Eq.~\eqref{eq:HT2} is stronger than the analysis provided Beigi and Gohari \cite[Theorem 6]{BG14} in view of the relation between the quantum hypothesis-testing divergence and the quantum information-spectrum divergence, Eq.~\eqref{eq:TH13}.
This again magnifies the fact that the pretty-good measurement yields a one-shot achievability bound on the \emph{Stein exponent} (i.e.~the maximal exponent of the type-II error provided that the type-I error is at most $\eps\in(0,1)$), and it can achieve the second-order asymptotics in the i.i.d~setting as well.
Note here that since the pretty-good measurement is sub-optimal, it incurs a cost $-\log \frac{1}{\delta}$  on the Stein exponent in the one-shot setting and a third-order term $-\frac12 \log n$ in the $n$-fold i.i.d.~scenario. Yet, it is still sufficient to achieve the moderate deviation asymptotics \cite{CTT2017} (i.e.~the inferior third-order term $-\frac12 \log n$ does not affect the moderate deviation expansion).

Next, we show that the pretty-good measurement can recover the \emph{quantum Hoeffding bound} \cite[\S 5.5]{ANS+08}.
Applying the upper bound on ``$\wedge$" (Fact~\ref{fact:upper} with $\alpha = 1-s \in (0,1)$) on Eq.~\eqref{eq:HT}, we obtain:
\begin{align} \notag 
	\begin{dcases}
		\alpha 
		\leq \Tr\left[ \rho \wedge \mu \sigma \right]
		\leq \mu^{1-\alpha } \Tr\left[ \rho^\alpha \sigma^{1-\alpha} \right],
		& \\
		\beta 
		\leq \mu^{-1} \Tr\left[ \rho \wedge \mu \sigma \right]
		\leq \mu^{-\alpha} \Tr\left[ \rho^\alpha \sigma^{1-\alpha} \right]
		. &\\
	\end{dcases}
\end{align}
Choosing $\mu = \mathrm{e}^{ \frac{\alpha-1}{\alpha}D_\alpha(\rho\,\Vert\,\sigma) + \frac{r}{\alpha}  }$ with the quantum Petz--R\'enyi divergence $D_\alpha$ introduced in Proposition~\ref{prop:error}, we arrive at the one-shot quantum Hoeffding bound: For all $r>0$ and $\alpha \in(0,1)$,
\begin{align} \notag
	\begin{dcases}
		\alpha \leq \mathrm{e}^{ - \frac{1-\alpha}{\alpha}\left(D_\alpha(\rho\,\Vert\,\sigma) - r\right) }, &\\
		\beta \leq \mathrm{e}^{-r}. &\\
	\end{dcases}
\end{align}
To our best knowledge, this is the first time that the quantum Hoeffding bound is achieved by using the pretty-good measurement.

\subsection{Entanglement-Assisted Classical Communication Over Quantum Channels} \label{sec:EA}	

In this section, we elaborate on how the achievability of entanglement-assisted (EA) classical communication \cite{BSS+99, BSS+02, Hol02, Dup10, DH13, MW14} follows from the proposed simple derivation in Section~\ref{sec:Main} in the same fashion.

\begin{defn}[Entanglement-assisted (EA) classical communication over quantum channels]
	Let $\mathscr{N}_{\A\to \B}$ be a quantum channel.
	\begin{enumerate}[1.]
		\item Alice holds a classical register $\mathsf{M}$ and quantum registers $\A$ and $\A'$,  and Bob holds  quantum registers $\B$ and $\R'$.
		
		\item A resource of an arbitrary state $\theta_{\R'\A'}$ is shared between Bob and Alice beforehand.
		
		\item For any (equiprobable) message $m\in\mathsf{M}$ Alice wanted to send, she performs an encoding quantum operation $\mathscr{E}_{\A'\to \A}^m$ on $\theta_{\R'\A'}$.
		
		\item The quantum channel $\mathscr{N}_{\A\to \B}$ is applied on Alice's quantum register $\A$ and outputs a state on Bob's quantum register $\B$.
		
		\item Bob performs a decoding measurement $\{  {\Pi}_{\R'\B}^m\}_{m\in\mathsf{M}}$ on registers $\R'$ and $\B$ to extract the sent message $m$.
	\end{enumerate}
	
	An $(M, \eps)$-EA-code for $\mathscr{N}_{\A\to \B}$ is a protocol such that $|\mathsf{M}| = M$ and  the average error probability satisfies
	\begin{align} \notag
		\frac{1}{M} \sum_{m\in \mathsf{M}}	\Tr\left[ \left( \mathds{1}-   {\Pi}_{\R'\B}^m \right) \mathscr{N}_{\A\to\B} \circ \mathscr{E}_{\A'\to \A}^m\left(\theta_{\R'\A'}\right)\right]
		\leq \varepsilon.
	\end{align}
\end{defn}

We adopt the encoder of the \emph{position-based coding} \cite{AJW19a} but with the pretty-good measurement as the decoder. 
\begin{itemize}
	\item \textbf{Preparations:} Alice and Bob pre-share an $M$-fold product state $\theta_{\R'\A'} := \theta_{\R\A}^{\otimes M} = \theta_{\R_1\A_1}\otimes \cdots \otimes \theta_{\R_M\A_M}$.
	
	\item \textbf{Encoding.} 
	For sending each $m\in\mathsf{M}$, 
	Alice simply sends her system $\A_m$, i.e.~$\mathscr{E}_{\A^M\to \A}^m = \Tr_{\A^{\mathsf{M}\backslash \{m\}}}$,
	for tracing out systems $\A_{\bar{m}}$ for all $\bar{m}\neq  m$. 
	
	\item \textbf{Decoding.} At receiver, the channel output states are 
	\begin{align} \label{eq:rho_m}
		\rho_{\R^M  \B}^m := \theta_{\R}^{\otimes (m-1)} \otimes \mathscr{N}_{\A\to\B}(\theta_{\R_m\A_m}) \otimes \theta_{\R}^{\otimes (M-m)}, \quad \forall\, m\in\mathsf{M}.
	\end{align}
	Then, Bob performs the pretty-good measurement with respect to the channel output states:
	\begin{align} \notag
		 {\Pi}_{\R^M\B}^m := \frac{ \rho_{\R^M  \B}^m }{ \sum_{\bar{m}\in\mathsf{M}} \rho_{\R^M  \B}^{\bar{m}} }, \quad \forall m\in\mathsf{M}.
	\end{align}
\end{itemize}

Note that the decoding part constitutes the main difference from the previous results, such as the original position-based coding \cite{AJW19a, QWW18} based on the Hayashi--Nagaoka operator inequality \cite[Lemma 2]{HN03}, the sequential decoding strategy \cite{Wil13} with an auxiliary probe system, and a quantum union bound \cite[Theorem 2.1]{OMW19} \footnote{The pretty-good measurement was also used in Ref.~\cite[\S 4]{DTW16} by following Beigi--Gohari's approach \cite{BG12}, wherein the obtained one-shot expression and the asymptotic analysis are more involved.}.

Below, we analyze the conditional error probability for sending each message $m\in\mathsf{M}$.
Let $\Tr_{\R^{\mathsf{M}\backslash \{m\}}}$ be the partial trace for tracing out systems $\R_{\bar{m}}$ for all $\bar{m}\neq m$, except $\R_m$.
By \eqref{eq:rho_m}, we have the following identities:
\begin{align} \notag
	\begin{dcases}
		\Tr_{\R^{\mathsf{M}\backslash \{m\}}}[\rho_{\R^M \B}^{m}] = \mathscr{N}_{\A\to\B}(\theta_{\R_m \A_m}) = \mathscr{N}_{\A\to\B}(\theta_{\R \A}), & \\
		\Tr_{\R^{\mathsf{M}\backslash \{m\}}}[\rho_{\R^M \B}^{\bar m}] = \theta_{\R_m}\otimes \mathscr{N}_{\A\to\B}(\theta_{\A_m}) = \theta_{\R}\otimes \mathscr{N}_{\A\to\B}(\theta_{\A}), & \forall\,\bar{m}\neq m. \\
	\end{dcases}
\end{align}
Then, the error probability conditioned on sending each message $m\in\mathsf{M}$ is 
\begin{align}
	\notag
	\Tr\left[ \rho_{\R^M \B}^{m} \frac{ \sum_{\bar{m}\neq m} \rho_{\R^M  \B}^{\bar m} }{\rho_{\R^M  \B}^{m} + \sum_{\bar{m}\neq m} \rho_{\R^M  \B}^{\bar m}} \right]
	&\overset{\text{(a)}}{\leq} \Tr\left[ \rho_{\R^M  B}^{m} \, \Exterior \left(\sum\nolimits_{\bar{m}\neq m} \rho_{\R^M  \B}^{\bar m} \right) \right] \\
	\notag
	&\overset{}{=} \Tr\left[\Tr_{\R^{\mathsf{M}\backslash \{m\}}}\left[   \rho_{\R^M  \B}^{m} \, \Exterior \left(\sum\nolimits_{\bar{m}\neq m} \rho_{\R^M  \B}^{\bar m} \right) \right] \right] \\	
	&\overset{\text{(b)}}{\leq} \Tr\left[ \left( \Tr_{\R^{\mathsf{M}\backslash \{m\}}}\left[ \rho_{\R^M  \B}^{m} \right] \right) \Exterior  \left(\sum\nolimits_{\bar{m}\neq m} \Tr_{\R^{\mathsf{M}\backslash \{m\}}}\left[ \rho_{\R^M  \B}^{\bar m} \right) \right] \right] \label{eq:EA_temp}\\	
	\notag
	&= \Tr\left[ \mathscr{N}_{\A\to\B}(\theta_{\R\A}) \wedge \left(M-1\right)  \theta_{\R}\otimes \mathscr{N}_{\A\to\B}(\theta_{\A})\right], 
\end{align}
where, likewise in the proof of Theorem~\ref{theo:main}, (a) follows from the lower bound of the noncommutative minimal, Fact~\ref{fact:lower};
(b) is due to the monotonicity of the  noncommutative minimal under positive trace-preserving maps, Fact~\ref{fact:monotone}.
Hence, we establish the following one-shot achievability for entanglement-assisted classical communication over quantum channels.

\begin{theo}[A one-shot achievability bound for EA classical communication over quantum channels] \label{theo:EA} 
	Consider an arbitrary quantum channel $\mathscr{N}_{\A \to \B}$. Then, there exists an $(M,\eps)$-EA-code for $\mathscr{N}_{\A \to \B}$ such that for any density operator $\theta_{\R\A}$,
	\begin{align} \notag
		\eps
		\leq \Tr\left[ \mathscr{N}_{\A\to\B}(\theta_{\R\A}) \wedge \left(M-1\right)  \theta_{\R}\otimes \mathscr{N}_{\A\to\B}(\theta_{\A})\right].
	\end{align}
\end{theo}

\begin{remark}
	The above derivations re-emphasize the central idea of the position-based coding proposed by Anshu, Jain, and Warsi \cite{AJW19a}. Namely, the pre-shared entanglement $\theta_{\R^M\A^M} = \theta_{\R\A}^{\otimes M}$ along with the encoding $m\mapsto \rho_{\R^M \B}^m = \mathscr{N}_{\A\to\B}(\theta_{\R^M\A_m})$
	ensure the \emph{mutual independence} between each subsystem $\R_m \A_m$, for all $m\in\mathsf{M}$, and accordingly, $\Tr_{\R^{\mathsf{M}\backslash \{m\}}}[\rho_{\R^M \B}^{\bar m}] = \theta_{\R}\otimes \mathscr{N}_{\A\to\B}(\theta_{\A})$ for all $\bar{m}\neq m$.
	Here, the partial trace $\Tr_{\R^{\mathsf{M}\backslash \{m\}}}$ may be considered as an \emph{expectation conditioned on $m$}  
	(see Remark~\ref{remark:conditional_expectation} for a detailed discussion).
	Such independence between the register $\R_m$ associated to each channel output state thus plays the same role as the independent random codebook used in classical-quantum channel coding (Theorem~\ref{theo:main}).
	On the other hand, we would like to point out that, normally, a communication system operates on a message set whose size is exponentially large, i.e.~$M\geq \e^{n I(\R{\,:\,}\B)_{\mathscr{N}(\theta)} - O(n) }$.
	Preparing exponentially many copies of the pre-shared state $\theta_{\R\A}^{\otimes M}$ might be practically challenging. 
	(Note that even in the classical case, resources required to generate mutual independence among an exponentially large set could not be considered as non-expensive \cite[Corollary 3.34]{Vad12}).
	Nevertheless, Anshu, Jain, and Warsi \cite[\S IV]{AJW19a} adapted an \emph{entanglement recycling} technique by Strelchuk, Horodecki, and Oppenheim \cite{SHO13} to reduce the required amount of entanglement resource.
	
	From our analysis given above, only \emph{pairwise independence} among each subsystem $\R_m \A_m$, $m\in\mathsf{M}$, is needed. That is, we only requires that $\Tr_{\R^{\mathsf{M}\backslash \{m,\bar{m}\}} \A^{\mathsf{M}\backslash \{m,\bar{m}\}}}[\theta_{\R^M\A^M} ] = \theta_{\R_m\A_m}\otimes \theta_{\R_{\bar{m}}\A_{\bar{m}}}$ for each $m\neq \bar{m}$.
	This point of view may provide another angle to reduce the required entanglement resource.
	Though, to our best knowledge, its explicit construction is not clear in noncommutative probability space.
	We leave this for future work.
\end{remark}

\begin{remark} \label{remark:conditional_expectation}
	The analysis of Theorem~\ref{theo:EA} actually shares the same flavor as that of Theorem~\ref{theo:main}.
	More precisely, the partial trace $\Tr_{\R^{\mathsf{M}\backslash \{m\}}}$ in Eq.~\eqref{eq:EA_temp} plays the same role as
	the averaging over the random codebook as in Eq.~\eqref{eq:main_temp}.
	In other words, the partial trace $\Tr_{\R^{\mathsf{M}\backslash \{m\}}}$ can be interpreted as a \emph{conditional expectation} \cite{Ume54, Ume56, Ume59, Car09} (which is a completely positive and trace-preserving map) from the operator algebra of bounded operators on $\mathsf{R}^M \B$, i.e.~$\mathcal{B}(\R^M  \B)$, to its subalgebra \footnote{Namely, the subalgebra consists of all operators in $\cB(\R^M  \B)$ of the form $\mathds{1}_{\R^{m-1}}\otimes  \Upsilon_{\R_m \B} \otimes \mathds{1}_{\R^{M-m}}$ for all $ \Upsilon_{\R_m \B} \in \cB(\R_m \B)$.} $\mathds{1}_{\R^{m-1}} \otimes \cB(\R_m\B) \otimes \mathds{1}_{\R^{M-m}}$.
	%
\end{remark}


Directly applying the pretty-good measurement as above allows us to obtain a tighter and cleaner one-shot achievability bound in a more general form. This then revisits the \emph{position-based coding} proposed by Anshu, Jain, and Warsi \cite[Lemma 4]{AJW19a}.
We summarize it as the following \emph{one-shot quantum packing lemma} that is not only prominent to Theorems~\ref{theo:main} and \ref{theo:EA}, and all the forthcoming results in this section, but we believe it is applicable elsewhere in quantum information theory as well.

\begin{theo}[A one-shot quantum packing lemma] \label{lemm:packing}
	Let $\rho_{\R\B}$ and $\tau_{\R}$ be arbitrary density operators and let $M$ be an integer.
	For every $m\in\mathsf{M} := \{1,\ldots,M\}$, define
	\begin{align} \notag 
		\omega_{\R_1 \R_2 \ldots \R_M \B}^m := \rho_{\R_m \B} \otimes \tau_{\R_1} \otimes \tau_{\R_2} \otimes \cdots \otimes \tau_{\R_{m-1}} \otimes \tau_{\R_{m+1}} \otimes \cdots \otimes \tau_{\R_M},
	\end{align}
	where $\rho_{\R_m \B} = \rho_{\R\B}$ and $\tau_{\R_m} = \tau_{\R}$ for every $m\in\mathsf{M}$.
	Then, there exists a measurement
	\begin{align} \notag
		 {\Pi}_{\R_1 \R_2 \ldots \R_M \B}^m := \frac{ \omega_{\R_1 \R_2 \ldots \R_M \B}^m }{ \sum_{\bar{m}\in\mathsf{M}} \omega_{\R_1 \R_2 \ldots \R_M \B}^{\bar{m}} }, \quad \forall m\in\mathsf{M},
	\end{align}
	satisfying, for every $m\in\mathsf{M}$,
	\begin{align} \notag 
		\Tr\left[\omega_{\R_1 \R_2 \ldots \R_M \B}^m \left( \mathds{1} -   {\Pi}_{\R_1 \R_2 \ldots \R_M \B}^m \right) \right] \leq \Tr\left[ \rho_{\R\B} \wedge (M-1) \tau_{\R} \otimes \rho_{\B} \right]. 
	\end{align}
\end{theo}

To see how the one-shot quantum packing lemma is applied to the previous achievability bounds, we make the following substitutions: 
$\rho_{\R_m \B} \to \mathscr{N}_{\A\to\B}(\theta_{\R_m\A_m})$,
and $\tau_{\R_{\bar{m}}} \to \theta_{\R_{\bar{m}}}$ for all $m\in\mathsf{M}$ and $\bar{m}\neq m$.
Then, Theorem~\ref{lemm:packing} covers Theorem~\ref{theo:EA} for entanglement-assisted classical communication over quantum channels.

On the other hand, in the scenario where 
$\R_m \B \to \X_m \B$,
$\rho_{\R_m \B} \to \rho_{\X_m \B}$,
and $\tau_{\R_{\bar{m}}} \to \rho_{\X_{\bar{m}}}$  for all $m\in \mathsf{M}$ and $\bar{m}\neq m$,
the setting in Theorem~\ref{lemm:packing} corresponds to the \emph{randomness-assisted communication over c-q channels}, 
where the systems $\X_{m}$'s are the shared randomness at Bob; and the joint state $\rho_{\X_m \B}$ resulted from Alice sending her $m$-th classical system through the c-q channel
(see also the paper by Wilde \cite{Wil17b} and Anshu--Jain--Warsi \cite{AJW19}).
Then, Theorem~\ref{lemm:packing} yields the achievability bound on the average error probability over the ensemble of codes, i.e.~the right-hand side of \eqref{eq:main} in Theorem~\ref{theo:main}.
Via de-randomization, one can always claim the existence of a good code in the ensemble to achieve such an error bound without randomness assistance.
This concludes the statement of Theorem~\ref{theo:main} for c-q channel coding~\footnote{We note that 
	by applying Theorem~\ref{lemm:packing} in randomness-assisted communication over c-q channels, it is equivalent to calculating the average error probability using \emph{mutually independent} random codebook, while in Theorem~\ref{theo:main}, we only require \emph{pairwise independent} random codebook.
}.

\medskip
Following the same reasoning as in Proposition~\ref{prop:error}, Theorem~\ref{theo:EA} (or Theorem~\ref{lemm:packing}) leads to a large deviation type bound, which is tighter than \cite[Theorem 6]{QWW18} without a prefactor $4$; 
following the same reasoning as in Proposition~\ref{prop:rate}, Theorem~\ref{theo:EA} provides a tighter lower bound on the \emph{$\varepsilon$-one-shot entanglement-assisted capacity} for $\mathscr{N}_{\A\to\B}$ (i.e.~the maximal logarithmic size of messages with average error probability below $\varepsilon$) than \cite[Theorem 1]{AJW19a}, \cite[Theorem 8]{QWW18}, and \cite[Theorem 5.1]{OMW19} (with the same improvements as the comparison made in Section~\ref{sec:comparison}).
\begin{prop}[Bounding coding error given fixed coding rate]
	Consider an arbitrary quantum channel $\mathscr{N}_{\A \to \B}$. 
	Then, for any $R>0$, there exists an $\left(\mathrm{e}^{R},\eps\right)$-EA-code for $\mathscr{N}_{\A \to \B}$ such that for any $\theta_{\R\A}$,
	\begin{align} \notag
		\eps \leq \e^{- \frac{1-\alpha}{\alpha} \left( I_{2-\sfrac{1}{\alpha}}^\downarrow (\R{\,:\,}\B)_{\mathscr{N}(\theta)} - R \right) }, \quad \forall \alpha \in (\sfrac12,1).
	\end{align}
	Here, we follow the notation given in Proposition~\ref{prop:error}.
\end{prop}

\begin{prop}[Bounding coding rate given fixed coding error]
	Consider an arbitrary quantum channel $\mathscr{N}_{\A \to \B}$. 
	Then, for any $\eps\in(0,1)$ there exists an $\left(M,\eps\right)$-EA-code for $\mathscr{N}_{\A \to \B}$ such that for any $\theta_{\R\A}$ and any $\delta \in (0,\eps)$,
	\begin{align} \notag
		\log M \geq D_\textnormal{h}^{\eps-\delta}(\mathscr{N}_{\A\to\B}(\theta_{\R\A})\,\|\,\theta_{\R}\otimes \mathscr{N}_{\A\to\B}(\theta_{\A}))  - \log \tfrac{1}{\delta}.
	\end{align}
	Here, we follow the notation given in Proposition~\ref{prop:rate}.
\end{prop}

\subsection{Classical Data Compression with Quantum Side Information} \label{sec:CQSW}

In this section, we show that how the proposed method in Section~\ref{sec:Main} can be applied to classical data compression with quantum side information \cite{DW03, renes_one-shot_2012, CHDH-2018, CHDH2-2018, Hao-Chung}. Subsequently, we will  term such an protocol as CQSW (which stands for \emph{classical-quantum Slepian--Wolf}).

\begin{defn}[Classical data compression with quantum side information]
	Let $\rho_{\X\B} = \sum_{x\in\X} p_{\X} (x) |x\rangle \langle x| \otimes \rho_{\B}^x$ be a classical-quantum state.
	\begin{enumerate}[1.]
		\item Alice holds classical registers $\X$ and $\mathsf{M}$, and Bob holds a quantum register $\B$.
		
		
		\item Alice performs an encoding $\mathscr{E}:\X \to \mathsf{M}$ that compresses the source in $\X$ to an index in $\mathsf{M}$.
		
		\item Bob performs a decoding measurement described by a family of POVMs indexed by $m\in\mathsf{M}$, i.e.~$\{  {\Pi}_{\B}^{x,m} \}_{x\in\mathsf{X}}$ on register $\B$, to recover the source $x\in\X$.
	\end{enumerate}
	
	An $(M, \eps)$-CQSW-code for $\rho_{\X\B}$ is a protocol such that  $|\mathsf{M}| = M$ and the error probability satisfies
	\begin{align} \notag
		\sum_{ x\in\X }	p_{\X}(x) \Tr\left[ \rho_{\B}^{x} \left(\mathds{1}_{\B} -   {\Pi}_{\B}^{x, \mathscr{E}(x)} \right) \right]
		\leq \varepsilon.
	\end{align}
\end{defn}

Without loss of generality, we assume that the prior distribution of the source, $p_{\X}$, has full support for brevity.
We also adopt the standard random coding strategy as in Section~\ref{sec:Main}.
\begin{itemize}
	\item \textbf{Encoding.} 
	The encoder maps each $x\in\X$ pairwise independently to uniform index $m\in\mathsf{M}$.
	
	\item \textbf{Decoding.} We use the following pretty-good measurement (again given the realization of the above encoding):
	\begin{align} \label{eq:PGM_CQSW}
		 {\Pi}_{\B}^{x, m} := \frac{ p_{\X}(x) \rho_{\B}^x }{ \sum_{\bar{x}: \mathscr{E}(\bar{x}) = m } p_{\X}(\bar{x}) \rho_{\B}^{\bar{x}} }, \quad \forall x\in \X, \, m \in\mathsf{M}.
	\end{align}
\end{itemize}


\begin{theo}[A one-shot achievability bound for classical data compression with quantum side information] 
	Consider an arbitrary classical-quantum state $\rho_{\X\B} = \sum_{x\in\X} p_{\X} (x) |x\rangle \langle x| \otimes \rho_{\B}^x$. Then, there exists an $(M,\eps)$-CQSW-code for $\rho_{\X\B}$ such that
	\begin{align} \notag
		\eps
		\leq \Tr\left[ \rho_{\X\B} \wedge \tfrac{1}{M}  \mathds{1}_{\X}\otimes \rho_{\B}\right].
	\end{align}
\end{theo}
\begin{proof}
	We use the pretty-good measurement given in Eq.~\eqref{eq:PGM_CQSW} to calculate the expected error probability (over the random encoding) as follows:
	\begin{align*}
		\mathds{E}_{x\sim p_{\X}} \mathds{E}_{m \sim \frac{1}{M} } \Tr\left[ \rho_{\B}^x \frac{ \sum_{\bar{x}\neq x, \, \mathscr{E}(\bar{x})= m } p_{\X}(\bar{x}) \rho_{\B}^{\bar{x}} }{ \sum_{\bar{x}: \mathscr{E}(\bar{x}) = m } p_{\X}(\bar{x}) \rho_{\B}^{\bar{x}} } \right]
		&\overset{\text{(a)}}{\leq}
		\mathds{E}_{x\sim p_{\X}} \mathds{E}_{m \sim \frac{1}{M} } \Tr\left[ \rho_{\B}^x  \,\Exterior \left( \sum_{\bar{x}\neq x, \, \mathscr{E}(\bar{x})= m } \frac{p_{\X}(\bar{x})}{ p_{\X}({x}) } \rho_{\B}^{\bar{x}} \right) \right]\\
		&\overset{\text{(b)}}{\leq}
		\mathds{E}_{x\sim p_{\X}} \Tr\left[ \rho_{\B}^x  \,\Exterior\,\mathds{E}_{ m\sim \frac{1}{M} } \left[ \sum_{\bar{x}\neq x } \mathbf{1}_{\{ \mathscr{E}(\bar{x}) = m \} } \frac{p_{\X}(\bar{x})}{ p_{\X}({x}) } \rho_{\B}^{\bar{x}} \right] \right]\\
		&\overset{\text{(c)}}{=} \mathds{E}_{x\sim p_{\X}} \Tr\left[ \rho_{\B}^x  \,\Exterior \left(\frac{1}{M} \sum_{\bar{x}\neq x } \frac{p_{\X}(\bar{x})}{ p_{\X}({x}) }  \rho_{\B}^{\bar{x}} \right) \right] \\
		&\overset{\text{(d)}}{\leq} \mathds{E}_{x\sim p_{\X}} \Tr\left[ \rho_{\B}^x \,\Exterior \left( \frac{1}{M p_{\X} (x)} \rho_{\B} \right)\right] \\
		&\overset{\text{(e)}}{\leq} \Tr\left[ \rho_{\X\B}  \,\Exterior \left(\frac{1}{M}  \mathds{1}_{\X}\otimes \rho_{\B} \right)\right],
	\end{align*}
	where (a) uses the lower bound of the noncommutative minimal given in Fact~\ref{fact:lower};
	(b) follows from the concavity given in Fact~\ref{fact:concave};
	(c) follows from the pairwise-independent and uniform random encoding;
	(d) follows from the monotone increase in the Loewner ordering and $\sum_{\bar{x}\neq x} p_{\X}(\bar{x})\rho_{\B}^{\bar{x}} \leq \sum_{\bar{x}} p_{\X}(\bar{x})\rho_{\B}^{\bar{x}} = \rho_{\B}$;
	and lastly, (e) follows from the direct sum formula given in Fact~\ref{fact:direct}.
\end{proof}

Using the same reasoning as in Propositions~\ref{prop:error} and \ref{prop:rate} of Section~\ref{sec:Main}, we have the following one-shot bounds for CQSW.
\begin{prop}[Bounding coding error given fixed coding rate]
	Consider an arbitrary classical-quantum state $\rho_{\X\B} = \sum_{x\in\X} p_{\X} (x) |x\rangle \langle x| \otimes \rho_{\B}^x$.
	Then, for any $R>0$, there exists an $\left(\mathrm{e}^{R},\eps\right)$-CQSW-code for $\rho_{\X\B}$ such that
	\begin{align} \notag
		\eps \leq \e^{- \frac{1-\alpha}{\alpha} \left( R -  H_{2-\sfrac{1}{\alpha}}^\downarrow (\X{\,|\,}\B)_{\rho} \right) }, \quad \forall \alpha \in (\sfrac12,1),
	\end{align}
	where $H_\alpha^\downarrow (\X{\,|\,}\B)_{\rho}  := - D_\alpha (\rho_{\X\B}\,\|\,\mathds{1}_{\X}\otimes \rho_{\B})$.
\end{prop}

\begin{prop}[Bounding coding rate given fixed coding error]
	Consider an arbitrary classical-quantum state $\rho_{\X\B} = \sum_{x\in\X} p_{\X} (x) |x\rangle \langle x| \otimes \rho_{\B}^x$.
	Then, for any $\eps\in(0,1)$ there exists an $\left(M,\eps\right)$-CQSW-code for $\rho_{\X\B}$ such that for any $\delta \in (0,\eps)$,
	\begin{align} \notag
		\log M \leq H_\textnormal{h}^{\eps-\delta}\left(\X {\,|\,} \B\right)_\rho  + \log \tfrac{1}{\delta},
	\end{align}
	where $H_\textnormal{h}^{\eps-\delta}\left(\X {\,|\,} \B\right)_\rho :=  - D_\textnormal{h}^{\eps-\delta}(\rho_{\X\B} \,\|\, \mathds{1}_{\X}\otimes \rho_{\B} )$.
\end{prop}

\subsection{Multiple-Access Channel Coding} \label{sec:MAC}
In this section, we show one-shot achievability bounds for classical-quantum multiple-access channel (MAC) coding and entanglement-assisted classical communication over quantum MACs \cite{Win01, HDW08, XW12, QWW18}. Note that the former naturally extends to (unassisted) classical communication over quantum MACs.
We will present the scenario for only two senders with one receivers; the result applies to multiple senders in the same fashion. 

\begin{defn}[Classical-quantum multiple-access channel coding]
	Let $\mathscr{N}_{\X\Y\to \C}: (x,y)\mapsto \rho_{\C}^{x,y}$ be a classical-quantum multiple-access channel.
	\begin{enumerate}[1.]
		\item Alice holds classical registers $\mathsf{M}_{\A}$ and $\X$, Bob holds $\mathsf{M}_{\B}$ and $\Y$, and Charlie holds quantum register $\C$.
		
		
		\item Alice performs an encoding $m_{\A} \mapsto x(m_{\A}) \in \X$ for any equiprobable message $m_{\A} \in\mathsf{M}_{\A}$ she wanted to send.
		Bob performs an encoding $m_{\B} \mapsto y(m_{\B}) \in \Y$ for any equiprobable message $m_{\B} \in \mathsf{M}_{\B}$ he wanted to send.
		
		\item The channel $\mathscr{N}_{\X\Y\to \C}$ is applied on Alice and Bob's registers $\X$ and $\Y$ and outputs a state on $\C$ at Charlie.
		
		\item A decoding measurement $\{  {\Pi}_{\C}^{m_{\A}, m_{\B}}\}_{(m_{\A}, m_{\B})\in\mathsf{M}_{\A} \times \mathsf{M}_{\B} }$ is performed on register $ \C$ to extract the sent message $(m_{\A}, m_{\B})$.
	\end{enumerate}
	
	An $(M_{\A}, M_{\B}, \eps)$-code for $\mathscr{N}_{\X\Y\to \C}$ is a protocol such that  $|\mathsf{M}_{\A}| = M_{\A}$,  $|\mathsf{M}_{\B}| = M_{\B}$, and the average error probability satisfies
	\begin{align} \notag
		\frac{1}{M_{\A} M_{\B}} \sum_{ (m_{\A}, m_{\B})\in\mathsf{M}_{\A} \times \mathsf{M}_{\B} }	\Tr\left[ \left( \mathds{1}-   {\Pi}_{\C}^{m_{\A}, m_{\B}} \right) \rho_{\C}^{x(m_{\A}), y(m_{\B})} \right]
		\leq \varepsilon.
	\end{align}
\end{defn}

We follow the strategy presented in Section~\ref{sec:Main}. 
\begin{itemize}
	\item \textbf{Encoding.} Each pair of messages $(m_{\A}, m_{\B})\in\mathsf{M}_{\A} \times \mathsf{M}_{\B}$ is mapped to a codeword $\left(x(m_{\A}), y(m_{\B}) \right) \in \X\times \Y$ pairwise independently according to some probability distribution $p_{\X} \otimes p_{\Y}$.
	
	\item \textbf{Decoding.} We use the pretty-good measurement with respect to the corresponding channel output states (given the realization of the random codebook):
	\begin{align} \notag
		 {\Pi}_{ \C}^{m_{\A}, m_{\B}} = \frac{ \rho_{\C}^{x(m_{\A}), y(m_{\B}) } }{ \sum_{(\bar{m}_{\A}, \bar{m}_{\B})} \rho_{ \C}^{x(\bar{m}_{\A}), y(\bar{m}_{\B})} }, \quad (m_{\A}, m_{\B})\in\mathsf{M}_{\A} \times \mathsf{M}_{\B}  .
	\end{align}
\end{itemize}

Then, we obtain the following result (without duplicating the proof).
\begin{theo}[A one-shot achievability bound for classical-quantum MAC coding] \label{theo:MAC_cq}
	Consider an arbitrary classical-quantum multiple-access channel $\mathscr{N}_{\X\Y\to \C}: (x,y)\mapsto \rho_{\C}^{x,y}$. Then, there exists an $\left(M_{\A}, M_{\B},\eps\right)$-code for $\mathscr{N}_{\X \Y\to \C}$ such that for any probability distributions $p_{\X}$ and $p_{\Y}$,
	\begin{align}  \notag
		\varepsilon \leq \Tr\left[ \rho_{\X\Y\mathsf{C}} \,\Exterior \left(
		(M_{\A}-1) \rho_{\X}\otimes \rho_{\Y\mathsf{C}}
		+ (M_{\B}-1) \rho_{\Y}\otimes \rho_{\X\mathsf{C}}
		+ (M_{\A}-1)(M_{\B}-1) \rho_{\X}\otimes \rho_{\Y} \otimes \rho_{\mathsf{C}}
		\right)	\right],
	\end{align}
	where $\rho_{\X\Y\mathsf{C}} := \sum_{(x,y)\in \X\times \Y} p_{\X}(x) |x\rangle\langle x|_{\X} \otimes p_{\Y}(y) |y\rangle\langle y|_{\Y} \otimes \rho_{\C}^{x,y}$.
\end{theo}

\medskip
Next, we consider the entanglement-assisted setting.

\begin{defn}[Entanglement-assisted classical communication over quantum multiple-access channels]
	Let $\mathscr{N}_{\A\B\to \C}$ be a quantum multiple-access channel.
	\begin{enumerate}[1.]
		\item Alice holds classical register $\mathsf{M}_{\A}$ and quantum registers $\A$ and $\A'$.
		Bob holds classical register $\mathsf{M}_{\B}$ and quantum registers $\B$ and $\B'$.
		Charlie holds quantum registers $\C$, $\R_{\A}'$ and $\R_{\B}'$.
		
		\item 
		Charlie and Alice share an arbitrary state $\theta_{\R_{\A}'\A'}$.
		Charlie and Bob share an arbitrary state $\theta_{\R_{\B}'\B'}$.
		
		\item Alice performs an encoding $\mathscr{E}_{\A'\to \A}^{m_{\A}}$ on  $\theta_{\R_{\A}'\A'}$ for any equiprobable message $m_{\A} \in\mathsf{M}_{\A}$ she wanted to send;
		Bob performs an encoding $\mathscr{E}_{\B'\to \B}^{m_{\B}}$ on $\theta_{\R_{\B}'\B'}$ for any equiprobable message $m_{\B} \in \mathsf{M}_{\B}$ he wanted to send.
		
		\item The channel $\mathscr{N}_{\A\B\to \C}$ is applied on Alice and Bob's registers $\A$ and $\B$, and outputs a state on quantum register $\C$ at Charlie.
		
		\item Charlie performs a decoding measurement $\{  {\Pi}_{\R_{\A}' \R_{\B}' \C}^{m_{\A}, m_{\B}}\}_{(m_{\A}, m_{\B})\in\mathsf{M}_{\A} \times \mathsf{M}_{\B} }$ on registers $\R_{\A}' \R_{\B}' \C$ to extract the sent message $(m_{\A}, m_{\B})$.
	\end{enumerate}
	
	An $(M_{\A}, M_{\B}, \eps)$-EA-code for $\mathscr{N}_{\A\B\to \C}$ is a protocol such that $|\mathsf{M}_{\A}| = M_{\A}$,  $|\mathsf{M}_{\B}| = M_{\B}$, and the average error probability satisfies
	\begin{align} \notag
		\frac{1}{M_{\A} M_{\B}} \sum_{ (m_{\A}, m_{\B})\in\mathsf{M}_{\A} \times \mathsf{M}_{\B} }	\Tr\left[ \left( \mathds{1}-   {\Pi}_{\R_{\A}' \R_{\B}' \C}^{m_{\A}, m_{\B}} \right) \mathscr{N}_{\A\B\to\C} \left( \mathscr{E}_{\A'\to \A}^{m_{\A}}\left( \theta_{\R_{\A}' \A'} \right) \otimes \mathscr{E}_{\B'\to \B}^{m_{\B}}\left( \theta_{\R_{\B}' \B'} \right) \right)\right]
		\leq \varepsilon.
	\end{align}
\end{defn}

As in Section~\ref{sec:EA}, we use the encoder of the position-based coding (see also \cite{QWW18}) and the pretty-good measurement for decoding.
\begin{itemize}
	\item \textbf{Preparation.} $M_{\A}$-fold product states $\theta_{\R_{\A}'\A'} := \theta_{\R_{\A}\A}^{\otimes M_{\A}}$ are shared between Charlie and Alice, and 
	$M_{\B}$-fold product states $\theta_{\R_{\B}'\B'} := \theta_{\R_{\B}\B}^{\otimes M_{\B}}$ are shared between  Charlie and Bob.
	
	\item \textbf{Encoding.} Alice adopts encoding $\mathscr{E}_{\A^{M_{\A}}\to \A}^{m_{\A}} = \Tr_{\A^{\mathsf{M}_{\A}\backslash \{m_{\A}\}}}$ for each $m_{\A}\in\mathsf{M}_{\A}$, and Bob adopts encoding $\mathscr{E}_{\B^{M_{\B}}\to \B}^{m_{\B}} = \Tr_{\B^{\mathsf{M}_{\B}\backslash \{m_{\B}\}}}$ for each $m_{\B}\in\mathsf{M}_{\B}$.
	
	\item \textbf{Decoding.} Denote the corresponding channel output state for sending message $(m_{\A}, m_{\B})\in\mathsf{M}_{\A} \times \mathsf{M}_{\B} $ by 
	\begin{align} \notag
		\rho_{\R_{\A}^{M_{\A}} \R_{\B}^{M_{\B}} \C}^{m_{\A}, m_{\B}} := \theta_{\R_{\A}}^{\otimes (m_{\A}-1)} \otimes \theta_{\R_{\B}}^{\otimes (m_{\B}-1)} \otimes \mathscr{N}_{\A\B\to\C}\left(\theta_{\R_{\A}\A}\otimes \theta_{\R_{\B}\B}\right) \otimes \theta_{\R_{\A}}^{\otimes (M_{\A}-m_{\A})} \otimes \theta_{\R_{\B}}^{\otimes (M_{\B}-m_{\B})}.
	\end{align}
	Then, we use the associated pretty-good measurement:
	\begin{align} \notag
		 {\Pi}_{\R_{\A} \R_{\B} \C}^{m_{\A}, m_{\B}} = \frac{ \rho_{\R_{\A}^{M_{\A}} \R_{\B}^{M_{\B}} \C}^{m_{\A}, m_{\B}} }{ \sum_{(\bar{m}_{\A}, \bar{m}_{\B})} \rho_{\R_{\A}^{M_{\A}} \R_{\B}^{M_{\B}} \C}^{\bar{m}_{\A}, \bar{m}_{\B}} }, \quad (m_{\A}, m_{\B})\in\mathsf{M}_{\A} \times \mathsf{M}_{\B}  .
	\end{align}
\end{itemize}

Following the analysis presented in Section~\ref{sec:EA}, we immediately obtain the following result (without duplicating the proof).
\begin{theo}[A one-shot achievability bound for EA classical communication over quantum MAC] 
	Consider an arbitrary quantum multiple-access channel $\mathscr{N}_{\A \B\to \C}$. Then, there exists an $\left(M_{\A}, M_{\B},\eps\right)$-EA-code for $\mathscr{N}_{\A \B\to \C}$ such that for any $\theta_{\R_{\A}\A}$ and $\theta_{\R_{\B}\B}$,
	\begin{align}  \notag
		\varepsilon \leq \Tr\left[ \rho_{\R_{\A}\R_{\B}\mathsf{C}} \,\Exterior \left(
		(M_{\A}-1) \rho_{\R_{\A}}\otimes \rho_{\R_{\B}\mathsf{C}}
		+ (M_{\B}-1) \rho_{\R_{\B}}\otimes \rho_{\R_{\A}\mathsf{C}}
		+ (M_{\A}-1)(M_{\B}-1) \rho_{\R_{\A}}\otimes \rho_{\R_{\B}} \otimes \rho_{\mathsf{C}}
		\right)	\right],
	\end{align}
	where $\rho_{\R_{\A}\R_{\B}\mathsf{C}} :=\mathscr{N}_{\!\A\B\to\mathsf{C}}\left(\theta_{\R_{\A}\A}\otimes \theta_{\R_{\B}\B}\right)$.
\end{theo}

\subsection{Broadcast Channel Coding} \label{sec:broadcast}
In this section, we study entanglement-assisted and unassisted classical communication over quantum broadcast channels \cite{YHD11, SW15, Dup10, DHL10, RSW16, WDW17, AJW19a, AJW19, Cheng2021b}.

\begin{defn}[Classical-quantum broadcast channel coding]
	Let $\mathscr{N}_{\X\to \B\C}: x\mapsto \rho_{\B\C}^x$ be a classical-quantum broadcast channel.
	\begin{enumerate}[1.]
		\item Alice holds classical registers $\mathsf{M}_{\B}$, $\mathsf{M}_{\C}$, and $\X$.
		Bob holds quantum register $\B$ and Charlie holds quantum register $\C$.
		
		
		\item Alice performs an encoding $(m_{\B}, m_{\C}) \mapsto x(m_{\B}, m_{\C}) \in \X$ for any equiprobable message $(m_{\B}, m_{\C})\in\mathsf{M}_{\B} \times \mathsf{M}_{\C}$ she wanted to send to Bob and Charlie, respectively.
		
		\item The channel $\mathscr{N}_{\X\to \B\C}$ is applied on Alice's register $\X$, and outputs a marginal state on $\B$ at Bob and a marginal state on $\C$ at Charlie.
		
		\item Bob performs a decoding measurement $\{  {\Pi}_{\B}^{ m_{\B}}\}_{ m_{\B}\in \mathsf{M}_{\B} }$ on register $ \B$ to extract the sent message $m_{\B}$, and Charlie performs a decoding measurement $\{  {\Pi}_{\C}^{ m_{\C}}\}_{ m_{\C}\in \mathsf{M}_{\C} }$ on register $ \C$ to extract the sent message $m_{\C}$.
	\end{enumerate}
	
	An $(M_{\B}, M_{\C}, \eps_{\B}, \eps_{\C})$-code for $\mathscr{N}_{\X\to \B\C}: x\mapsto \rho_{\B\C}^{x}$ is a protocol such that $|\mathsf{M}_{\B}| = M_{\B}$,  $|\mathsf{M}_{\C}| = M_{\C}$, and the average error probabilities satisfy
	\begin{align} \notag
		\begin{dcases}
			\frac{1}{ M_{\B}M_{\C} } \sum_{ (m_{\B}, m_{\C})\in\mathsf{M}_{\B} \times \mathsf{M}_{\C} }	\Tr\left[ \left( \mathds{1}-   {\Pi}_{\B}^{ m_{\B}} \right) \rho_{\B}^{x(m_{\B},m _{\C})} \right]
			\leq \varepsilon_{\B}; &\\
			\frac{1}{ M_{\B}M_{\C} } \sum_{ (m_{\B}, m_{\C})\in\mathsf{M}_{\B} \times \mathsf{M}_{\C} }	\Tr\left[ \left( \mathds{1}-   {\Pi}_{\C}^{ m_{\C}} \right) \rho_{\C}^{x(m_{\B},m_{\C})} \right]
			\leq \varepsilon_{\C}. &\\
		\end{dcases}	
	\end{align}
\end{defn}

We follow the analysis proposed in Section~\ref{sec:Main} by considering communication from Alice to Bob and from Alice to Charlie, separately.
\begin{itemize}
	\item \textbf{Encoding.} 
	We introduce two auxiliary classical registers $\U$ and $\V$ for precoding.
	Message $m_{\B} \in \mathsf{M}_{\B}$ for Bob is encoded to a pre-codeword $u(m_{\B}) \in \U$ pairwise independently  according to some probability distribution $p_{\U}$; message $m_{\C}\in \mathsf{M}_{\C}$ for Charlie is  encoded to a pre-codeword $v(m_{\C}) \in \V$ pairwise independently according to some probability distribution $p_{\V}$.
	Then, Alice picks a (deterministic) encoding $(u(m_{\B}), v(m_{\C})) \mapsto x(u(m_{\B}), v(m_{\C})) \in \X$.
	
	\item \textbf{Decoding.} 
	Denoting (with a slight abuse of notation) the marginal channel output states at Bob by 
	$\rho_{\B}^{m_{\B}} := \frac{1}{M_{\C}} \sum_{m_{\C} \in \mathsf{M}_{\C}} \rho_{\B}^{x(u(m_{\B}), v(m_{\C}))}$, $m_{\B} \in \mathsf{M}_{\B}$, Bob performs the corresponding pretty-good measurement:
	\begin{align} \notag
		 {\Pi}_{\B}^{ m_{\B}} := \frac{ \rho_{\B}^{m_{\B}} }{ \sum_{\bar{m}_{\B} \in \mathsf{M}_{\B} } \rho_{\B}^{\bar{m}_{\B}} }, \quad m_{\B} \in \mathsf{M}_{\B}.
	\end{align}
	Similarly, denoting (with a slight abuse of notation again) the marginal channel output states at Charlie by 
	$\rho_{\C}^{m_{\C}} := \frac{1}{M_{\B}} \sum_{m_{\B} \in \mathsf{M}_{\B}} \rho_{\C}^{x(u(m_{\B}), v(m_{\C}))}$, $m_{\C} \in \mathsf{M}_{\C}$, Charlie performs the corresponding pretty-good measurement:
	\begin{align} \notag
		 {\Pi}_{\C}^{ m_{\C}} := \frac{ \rho_{\C}^{m_{\C}} }{ \sum_{\bar{m}_{\C} \in \mathsf{M}_{\C} } \rho_{\C}^{\bar{m}_{\C}} }, \quad m_{\C} \in \mathsf{M}_{\C}.
	\end{align}
\end{itemize}

Then, we obtain the following result (without duplicating the proof).
\begin{theo}[A one-shot achievability bound for classical-quantum broadcast channel coding]\label{theo:broadcast_cq} 
	Consider an arbitrary classical-quantum broadcast channel $\mathscr{N}_{\X\to \B\C}: x \mapsto \rho_{\B\C}^{x}$. Then, there exists an $\left(M_{\B}, M_{\C},\eps_{\B}, \eps_{\C}\right)$-code for $\mathscr{N}_{\X\to \B\C}$ such that for any probability distributions $p_{\U}$ and $p_{\V}$, and (deterministic) encoding $(u,v)\mapsto x(u,v)$,
	\begin{align}  \notag
		\varepsilon_{\B} \leq \Tr\left[ \rho_{\U\mathsf{B}} \wedge 
		(M_{\B}-1) \rho_{\U}\otimes \rho_{\B} \right]; \\
		 \notag
		\varepsilon_{\C} \leq \Tr\left[ \rho_{\U\mathsf{C}} \wedge 
		(M_{\C}-1) \rho_{\V}\otimes \rho_{\C} \right],
	\end{align}
	where $\rho_{\U\V\B\mathsf{C}} := \sum_{(u,v)\in \U\times \V} p_{\U}(u) |u\rangle\langle u|_{\U} \otimes p_{\V}(v) |v\rangle\langle v|_{\V} \otimes \rho_{\B\C}^{x(u,v)}$.
\end{theo}
Note that Theorem~\ref{theo:broadcast_cq} extends to classical communication over quantum broadcast channels straightforwardly (see Table~\ref{table:Table}).

\begin{remark}
	Theorem~\ref{theo:broadcast_cq} employs independent pre-coding $p_{\U} \otimes p_{\V}$ and hence it provides a simple and clean one-shot achievability bound. We note that such a scenario was considered by Anshu, Jain, and Warsi \cite{AJW19}. Hence, Theorem~\ref{theo:broadcast_cq} improves on the achievability in Ref.~\cite[Theorem 13]{AJW19}.
	
	Generally, Alice can adopt a joint pre-coding $p_{\U\V}$, which is called \emph{Marton's inner bound} in the classical setting \cite{Mar79, GvdM81} (see also the studies in the quantum setting \cite{YHD11, SW15, Dup10, DHL10, RSW16, AJW19a}); however, it would require additional covering techniques.
	We leave this for future work \cite{CG22b}.
\end{remark}

\medskip
Next, we present entanglement-assisted classical communication over quantum broadcast channels.

\begin{defn}[Entanglement-assisted classical communication over quantum broadcast channels]
	Let $\mathscr{N}_{\A\to \B\C}$ be a quantum broadcast channel.
	\begin{enumerate}[1.]
		\item 
		Alice holds classical registers $\mathsf{M}_{\B}$ and $\mathsf{M}_{\C}$, and
		quantum registers $\A_{\B}'$ and $\A_{\C}'$.
		Bob holds quantum registers $\B$ and $\R_{\B}'$.
		Charlie holds quantum registers $\C$ and $\R_{\C}'$.
		
		\item 
		Bob and Alice share an arbitrary state $\theta_{\R_{\B}'\A_{\B}'}$.
		Charlie and Alice share an arbitrary state $\theta_{\R_{\C}'\A_{\C}'}$.
		
		\item Alice performs an encoding $\mathscr{E}^{m_{\B}, m_{\C}}_{\A_{\B}' \A_{\C}' \to \A}$ on $\theta_{\R_{\B}'\A_{\B}'} \otimes \theta_{\R_{\C}'\A_{\C}'}$
		for any equiprobable message $(m_{\B}, m_{\C})\in\mathsf{M}_{\B} \times \mathsf{M}_{\C}$ she wanted to send to Bob and Charlie, respectively.
		
		\item The channel $\mathscr{N}_{\A\to \B\C}$ is applied on Alice's register $\A$, and outputs a marginal state on $\B$ at Bob and a marginal state on $\C$ at Charlie.
		
		\item Bob performs a decoding measurement $\{  {\Pi}_{\R_{\B}'\B}^{ m_{\B}}\}_{ m_{\B}\in \mathsf{M}_{\B} }$ on register $ \B$ to extract the sent message $m_{\B}$, and Charlie performs a decoding measurement $\{  {\Pi}_{\R_{\C}'\C}^{ m_{\C}}\}_{ m_{\C}\in \mathsf{M}_{\C} }$ on register $ \C$ to extract the sent message $m_{\C}$.
	\end{enumerate}
	
	An $(M_{\B}, M_{\C}, \eps_{\B}, \eps_{\C})$-EA-code for $\mathscr{N}_{\A\to \B\C}$ is a protocol such that $|\mathsf{M}_{\B}| = M_{\B}$,  $|\mathsf{M}_{\C}| = M_{\C}$ and the average error probabilities satisfy
	\begin{align}  \notag
		\begin{dcases}
			\frac{1}{ M_{\B}M_{\C} } \sum_{ (m_{\B}, m_{\C})\in\mathsf{M}_{\B} \times \mathsf{M}_{\C} }	\Tr\left[ \left( \mathds{1}-   {\Pi}_{\R_{\B}'\B}^{ m_{\B}} \right) \mathscr{N}_{\A\to\B\C} \circ \mathscr{E}^{m_{\B}, m_{\C}}_{\A_{\B}' \A_{\C}' \to \A} \left( \theta_{\R_{\B}'\A_{\B}'} \otimes \theta_{\R_{\C}'\A_{\C}'} \right) \right]
			\leq \varepsilon_{\B}; & \\
			\frac{1}{ M_{\B}M_{\C} } \sum_{ (m_{\B}, m_{\C})\in\mathsf{M}_{\B} \times \mathsf{M}_{\C} }	\Tr\left[ \left( \mathds{1}-   {\Pi}_{\R_{\C}'\C}^{ m_{\C}} \right)\mathscr{N}_{\A\to\B\C} \circ \mathscr{E}^{m_{\B}, m_{\C}}_{\A_{\B}' \A_{\C}' \to \A} \left( \theta_{\R_{\B}'\A_{\B}'} \otimes \theta_{\R_{\C}'\A_{\C}'} \right) \right]
			\leq \varepsilon_{\C}. &\\
		\end{dcases}	
	\end{align}
\end{defn}

We follow the similar analysis as above by considering communication from Alice to Bob and from Alice to Charlie, separately.
Again, we also employ the encoder of the position-based coding as in Refs.~\cite[Theorem 6]{AJW19a} and \cite[Theorem 6]{AJW19}, and apply the pretty-good measurement for decoding.
\begin{itemize}
	\item \textbf{Preparation.} 	
	Consider an arbitrary state $\theta_{\R_{\B} \R_{\C} \A}$ satisfying $\theta_{\R_{\B} \R_{\C}} = \theta_{\R_{\B}} \otimes \theta_{\R_{\C}}$. Let $\theta_{\R_{\B}\A_{\B} }$ be a purified state of $\theta_{\R_{\B}}$ and let $\theta_{\R_{\C} \A_{\C} }$ be a purified state of $\theta_{\R_{\C}}$.
	Then, Alice and Bob share the $M_{\B}$-fold product states
	$\theta_{\R_{\B}' \A_{\B}'} := \theta_{\R_{\B}\A_{\B} }^{\otimes M_{\B}}$,
	and Alice and Charlie share the $M_{\C}$-fold product states
	$\theta_{\R_{\C}' \A_{\C}'} := \theta_{\R_{\C}\A_{\C} }^{\otimes M_{\C}}$.
	
	\item \textbf{Encoding.} 
	For each $(m_{\B}, m_{\C} ) \in \mathsf{M}_{\B} \times \mathsf{M}_{\C}$, Alice sends the $(m_{\B}, m_{\C} )$-th registers $\A_{\B}$ and  $\A_{\C}$ and then performs an isometry transformation $\mathcal{V}_{\A_{\B}\A_{\C}\to \mathsf{E} \A}$ such that $\mathcal{V}_{\A_{\B}\A_{\C}\to \mathsf{E} \A}\left(\theta_{\R_{\B} \A_{\B}} \otimes \theta_{\R_{\C} \A_{\C}}\right)$ equals a purified state $\theta_{\mathsf{E}\R_{\B} \R_{\C} \A}$ of $\theta_{\R_{\B} \R_{\C} \A}$ with an additional purifying register $\mathsf{E}$.
	The overall encoding map is then $\mathscr{E}^{m_{\B}, m_{\C}}_{\A_{\B}^{M_{\B}} \A_{\C}^{M_{\C}} \to \A} = \Tr_{\mathsf{E}} \circ
	\mathcal{V}_{\A_{\B}\A_{\C}\to \mathsf{E} \A} \circ
	\Tr_{ \A_{\B}^{M_{\B}\backslash \{ m_{\B }\}} \A_{\C}^{M_{\C}\backslash \{ m_{\C }\}}  }$.
	%
	
	\item \textbf{Decoding.} 
	Denoting (with a slight abuse of notation) the marginal channel output states at Bob by 
	\begin{align} \notag
		\rho_{\R_{\B}^{M_\B} \B}^{m_{\B}} := \theta_{\R_{\B}}^{\otimes (m_{\B}-1)} \otimes \Tr_{\C\mathsf{E}\R_{\C} }  \circ
		\mathscr{N}_{\A\to\B\C} \circ \mathcal{V}_{\A_{\B}\A_{\C}\to \mathsf{E} \A} \left( \theta_{\R_{\B} \A_{\B}} \otimes \theta_{\R_{\C} \A_{\C}} \right) \otimes \theta_{\R_{\B}}^{\otimes (M_{\B} - m_{\B})}, \quad m_{\B} \in \mathsf{M}_{\B}.
	\end{align} 
	Bob performs the corresponding pretty-good measurement:
	\begin{align} \notag
		 {\Pi}_{\R_{\B}^{M_\B}\B}^{ m_{\B}} := \frac{ \rho_{\R_{\B}^{M_\B}\B}^{m_{\B}} }{ \sum_{\bar{m}_{\B} \in \mathsf{M}_{\B} } \rho_{\R_{\B}^{M_\B}\B}^{\bar{m}_{\B}} }, \quad m_{\B} \in \mathsf{M}_{\B}.
	\end{align}
	Similarly, denoting (with a slight abuse of notation again) the marginal channel output states at Charlie by 
	\begin{align} \notag
		\rho_{\R_{\C}^{M_\C} \C}^{m_{\C}} := \theta_{\R_{\C}}^{\otimes (m_{\C}-1)} \otimes \Tr_{\B\mathsf{E}\R_{\B} } \circ
		\mathscr{N}_{\A\to\B\C} \circ \mathcal{V}_{\A_{\B}\A_{\C}\to \mathsf{E} \A} \left( \theta_{\R_{\B} \A_{\B}} \otimes \theta_{\R_{\C} \A_{\C}} \right) \otimes \theta_{\R_{\C}}^{\otimes (M_{\C} - m_{\C})}, \quad m_{\C} \in \mathsf{M}_{\C}.
	\end{align} 
	Charlie performs the corresponding pretty-good measurement:
	\begin{align} \notag
		 {\Pi}_{\R_{\C}^{M_\C}\C}^{ m_{\C}} := \frac{ \rho_{\R_{\C}^{M_\C}\C}^{m_{\C}} }{ \sum_{\bar{m}_{\C} \in \mathsf{M}_{\C} } \rho_{\R_{\C}^{M_\C}\C}^{\bar{m}_{\C}} }, \quad m_{\C} \in \mathsf{M}_{\C}.
	\end{align}
\end{itemize}

Then, we obtain the following result (without duplicating the proof).
\begin{theo}[A one-shot achievability bound for EA classical communication over quantum broadcast channels]\label{theo:broadcast_EA} 
	Consider an arbitrary quantum broadcast channel $\mathscr{N}_{\A\to \B\C}$. Then, there exists an $\left(M_{\B},  M_{\C},\eps_{\B}, \eps_{\C}\right)$-EA-code for $\mathscr{N}_{\A\to \B\C}$ such that for any $\theta_{\R_{\B} \R_{\C} \A}$ satisfying $\theta_{\R_{\B} \R_{\C}} = \theta_{\R_{\B}} \otimes \theta_{\R_{\C}}$,
	\begin{align}  \notag
		\varepsilon_{\B} \leq \Tr\left[ \rho_{\R_{\B} \B} \wedge 
		(M_{\B}-1) \rho_{\R_{\B}}\otimes \rho_{\B} \right]; \\
		\notag
		\varepsilon_{\C} \leq \Tr\left[ \rho_{\R_{\C} \mathsf{C}} \wedge 
		(M_{\C}-1) \rho_{\R_{\C}}\otimes \rho_{\C} \right],
	\end{align}
	where $\rho_{\R_{\B} \R_{\C} \B\C} := \mathscr{N}_{\A\to\B\C} \left( \theta_{\R_{\B} \R_{\C} \A} \right) $.
\end{theo}

\subsection{Communication with Casual State Information at The Encoder} \label{sec:SI}
In this section, we consider entanglement-assisted and unassisted classical communication over quantum channels with causal channel state information available at the encoder \cite{Dup10, AJW19a, AJW19}, which is the quantum generalization of the classical Gel'fand--Pinsker channel \cite{GP80}; see also Ref.~\cite[\S 7]{GK11}.

\begin{defn}[Classical-quantum channel coding with causal state information]
	Let $\mathscr{N}_{\X\mathsf{S} \to \B}: (x,s)\mapsto \rho_{\B}^{x,s}$ be a  classical-quantum channel parameterized by $s\in\mathsf{S}$ and assume that a channel state $p_{\mathsf{S}}$ is available at the encoder.
	\begin{enumerate}[1.]
		\item 
		The channel holds a classical register $\mathsf{S}$.
		Alice holds classical registers $\mathsf{M}$, $\X$, and $\mathsf{S}'$ (an identical copy of $\mathsf{S}$).
		Bob holds a quantum register $\B$.
		
		
		\item 
		Given a realization of the channel state $s\in\mathsf{S}'$,
		an encoding $(m,s) \mapsto x(m,s)$ maps an equiprobable message $m \in \mathsf{M}$  to a codeword in $\X$.
		
		\item The classical-quantum channel $\mathscr{N}_{\X\mathsf{S}\to \B}$ is applied on Alice's register $\X$ given the realization of the channel state $s\in\mathsf{S}$, and outputs a state on $\B$ at Bob.
		(The realizations $s\in\mathsf{S}'$ at Alice and $s\in\mathsf{S}$ at the channel is identical.)
		
		\item Bob performs a decoding measurement described by a measurement $\{  {\Pi}_{\B}^{m} \}_{m\in\mathsf{M}}$ on $\B$ to extract the sent message $m\in\mathsf{M}$.
		(Note that Bob is aware of the mathematical description of the probability distribution $p_{\mathsf{S}}$ but not of the realization of a specific $s\in\mathsf{S}$.)
	\end{enumerate}
	
	An $(M, \eps)$-code for $\mathscr{N}_{\X\mathsf{S}\to \B}$ with state information $p_{\mathsf{S}}$ is a protocol such that  $|\mathsf{M}| = M$ and the average error probability satisfies
	\begin{align} \notag
		\frac{1}{M} \sum_{(m,s)\in \mathsf{M}\times \mathsf{S}} p_{\mathsf{S}}(s)	\Tr\left[ \rho_{\B}^{x(m,s),s} \left(\mathds{1}_{\B} -   {\Pi}_{\B}^{m} \right) \right]
		\leq \varepsilon.
	\end{align}
\end{defn}

We adopt the standard random coding strategy as follows.
\begin{itemize}
	\item \textbf{Encoding.} 
	We introduce an auxiliary classical register $\U$ for precoding.
	The message $m\in\mathsf{M}$ is mapped to a pre-codeword $u\in\U$ pairwise independently according to $p_{\U}$.
	With the realization of the channel state $s\in\mathsf{S}'$, Alice picks a (deterministic) encoding $(u(m), s) \mapsto x(u(m), s)\in\X$.
	
	\item \textbf{Decoding.} At the receiver, 
	denoting (with a slight abuse of notation)
	the channel output state by $\rho_{\B}^{m} := \sum_{s\in\mathsf{S}} p_{\mathsf{S}} (s) \rho_{\B}^{x(u(m),s), s}$ for each $m\in\mathsf{M}$,
	Bob performs the corresponding pretty-good measurement:
	\begin{align} \notag
		 {\Pi}^m_{\B} := \frac{\rho_{\B}^m }{\sum_{\bar{m}\in\mathsf{M}} \rho_{\B}^{\bar m}}, \quad \forall m\in\mathsf{M}.
	\end{align}
\end{itemize}

Then, following the analysis in Section~\ref{sec:Main}, we obtain a one-shot achievability bound (without duplicating the proof).
\begin{theo}[A one-shot achievability bound classical-quantum channel coding with casual state information] \label{theo:state_cq} 
	Consider an arbitrary classical-quantum channel $\mathscr{N}_{\X\mathsf{S} \to \B}: (x,s)\mapsto \rho_{\B}^{x,s}$ with state information $p_{\mathsf{S}}$.
	Then, there exists an $(M,\eps)$-code for $\mathscr{N}_{\X\mathsf{S} \to \B}$ with state information $p_{\mathsf{S}}$ such that for any probability distribution $p_{\U}$ and (deterministic) map $(u,s)\mapsto x(u,s)$,
	\begin{align} \notag
		\eps
		\leq \Tr\left[ \rho_{\U\B} \wedge (M-1) \rho_{\U}\otimes \rho_{\B} \right].
	\end{align}
	Here, $\rho_{\U\B} := \sum_{(u,s)\in\U\times \mathsf{S}} p_{\U}(u) |u\rangle \langle u| \otimes p_{\mathsf{S}}(s) \rho_{\B}^{x(u,s),s}$.
\end{theo}
The result extends to classical communication over quantum channels $\mathscr{N}_{\A\mathsf{S} \to \B}$ with quantum state information $\vartheta_{\mathsf{S}}$ (see also the definition below in the entanglement-assisted setting). We refer the reader to Table~\ref{table:Table} for the corresponding results.

\begin{remark}
	In the precoding phase of Theorem~\ref{theo:state_cq}, the chosen probability distribution $p_{\U}$ is independent of the channel state $p_{\mathsf{S}}$. This coding strategy is called \emph{casual state information} at the encoder \cite[\S 7.5]{GK11} and it was also studied in the quantum setting \cite[Theorem 12]{AJW19}.
	(Hence, Theorem~\ref{theo:state_cq} improves on \cite[Theorem 12]{AJW19}.)
	
	For the scenario of \emph{non-casual state information} at the encoder, the pre-coding probability distribution on $\U$ may be correlated with the channel state $p_{\mathsf{S}}$ \cite[\S 4]{Dup10}, \cite[\S V]{AJW19a}. This would require additional techniques. We leave this for future work \cite{CG22b}.
\end{remark}

\medskip

Next, we move on to the entanglement-assisted setting.
\begin{defn}[Entanglement-assisted classical communication over quantum channels with causal state information]
	Let $\mathscr{N}_{\A\mathsf{S} \to \B}$ be a quantum channel with a channel state $\vartheta_{\mathsf{S}}$.
	\begin{enumerate}[1.]
		\item 
		The channel holds a quantum register $\mathsf{S}$.
		Alice holds a classical register $\mathsf{M}$ and quantum registers $\A'$ and $\mathsf{S}'$.
		Bob holds quantum registers $\R'$ and $\B$.
		
		\item A resource of arbitrary state $\theta_{\R' \A'}$ is shared between Bob and Alice.
		Moreover, let $\vartheta_{\mathsf{S}'\mathsf{S}}$ be a purified state of $\vartheta_{\mathsf{S}}$ shared between Alice and the channel.
		
		\item Alice performs an encoding $\mathscr{E}^m_{\A'\mathsf{S}'\to \A}$ on registers $\A'$ and $\mathsf{S}'$ of state
		$\theta_{\R' \A'}\otimes \vartheta_{\mathsf{S}'\mathsf{S}}$ for any equiprobable message  $m \in \mathsf{M}$.
		
		\item The quantum channel $\mathscr{N}_{\A\mathsf{S}\to \B}$ is applied on Alice's register $\A$ and the register $\mathsf{S}$ of the channel state, and outputs a state on $\B$ at Bob.
		
		\item Bob performs a decoding measurement $\{  {\Pi}_{\R'\B}^{m} \}_{m\in\mathsf{M}}$ on $\R'\B$ to extract the sent message $m$.
		(Note that Bob is aware of the mathematical description of the channel state $\vartheta_{\mathsf{S}}$ but Bob cannot access to the channel's register $\mathsf{S}$ nor
		operate on such a channel state.)
	\end{enumerate}
	
	An $(M, \eps)$-EA-code for $\mathscr{N}_{\A\mathsf{S}\to \B}$ with state information $\vartheta_{\mathsf{S}}$ is a protocol such that $|\mathsf{M}| = M$ and the average error probability satisfies
	\begin{align} \notag
		\frac{1}{ M } \sum_{ m\in\mathsf{M} }	
		\Tr\left[ \left( \mathds{1}-   {\Pi}_{\R'\B}^{ m} \right) \mathscr{N}_{\A\mathsf{S}\to\B} \circ \mathscr{E}^{m}_{\A' \mathsf{S}' \to \A} \left(\theta_{\R' \A'}\otimes \vartheta_{\mathsf{S}'\mathsf{S}} \right) \right]
		\leq \varepsilon.
	\end{align}
\end{defn}

We also use the encoder of the position-based coding as in \cite[Theorem 5]{AJW19a}, \cite[Theorem 4]{AJW19}, and the pretty-good measurement for decoding.
\begin{itemize}
	\item \textbf{Preparation.} 	
	Consider an arbitrary state $\theta_{\R \A \mathsf{S}}$ satisfying $\theta_{\R \mathsf{S}} = \theta_{\R} \otimes \vartheta_{\mathsf{S}}$. 
	Let $\theta_{\R \U}$ be a purified state of $\theta_{\R}$ with an additional quantum register $\U$ at Alice.
	Then, Alice and Bob share $M$-fold product states
	$\theta_{\R' \A'} := \theta_{\R \U }^{\otimes M }$.
	
	\item \textbf{Encoding.} 
	For each $m  \in \mathsf{M}$, Alice sends the $m$-th register of $\U$ and applies an isometry transformation $\mathcal{V}_{\mathsf{S}' \U \to \mathsf{E} \A}$ such that $\mathcal{V}_{\mathsf{S}' \U \to \mathsf{E} \A}\left(\theta_{\R \U} \otimes \vartheta_{\mathsf{S}' \mathsf{S}} \right)$ equals a purified state $\theta_{\mathsf{E}\R \A \mathsf{S}} $ of $\theta_{\R \A \mathsf{S}}$ with an additional purifying register $\mathsf{E}$.
	Then, the overall encoding map is  $\mathscr{E}^{m}_{\U^{M}  \to \A} =  \Tr_{\mathsf{E} } \circ \mathcal{V}_{\mathsf{S}' \U \to \mathsf{E} \A} \circ
	\Tr_{ \U^{M \backslash \{ m\} } }$.
	
	\item \textbf{Decoding.}
	Denoting (with a slight abuse of notation again) the marginal channel output states at Bob by 
	\begin{align} \notag
		\rho_{\R^{M} \B}^{m} := \theta_{\R}^{\otimes (m-1)} \otimes 
		\mathscr{N}_{\A\mathsf{S}\to\B} \circ \Tr_{\mathsf{E} } \circ \mathcal{V}_{\mathsf{S}' \U \to \mathsf{E} \A} \left( \theta_{\R \U} \otimes \vartheta_{\mathsf{S}' \mathsf{S}} \right)  \otimes \theta_{\R}^{\otimes (M - m)}, \quad m \in \mathsf{M}.
	\end{align} 
	Bob performs the corresponding pretty-good measurement:
	\begin{align} \notag
		 {\Pi}_{\R^{M} \B}^{ m } := \frac{ \rho_{\R^{M}\B}^{m} }{ \sum_{\bar{m} \in \mathsf{M} } \rho_{\R^{M}\B}^{\bar{m}} }, \quad m \in \mathsf{M}.
	\end{align}
\end{itemize}

Then, applying the analysis in Section~\ref{sec:EA}, we obtain the following result (without duplicating the proof).
\begin{theo}[A one-shot achievability bound for EA classical communication over quantum channels with casual state information] 
	Consider an arbitrary quantum channel $\mathscr{N}_{\A\mathsf{S} \to \B}$ with state information $\vartheta_{\mathsf{S}}$.
	Then, there exists an $(M,\eps)$-EA-code for $\mathscr{N}_{\A\mathsf{S} \to \B}$ with state information $\vartheta_{\mathsf{S}}$
	such that for any $\theta_{\R\A\mathsf{S}}$ satisfying $\theta_{\R\mathsf{S}} = \theta_{\R} \otimes \vartheta_{\mathsf{S}}$,
	\begin{align} \notag
		\eps
		\leq \Tr\left[ \rho_{\R \B } \wedge (M-1) \rho_{\R}\otimes \rho_{\B} \right].
	\end{align}
	Here, $\rho_{\R \B } := \mathscr{N}_{\A\mathsf{S}\to\B } \left( \theta_{\R\A\mathsf{S}} \right) $.
\end{theo}

\section{Conclusions} \label{sec:conclusions}

We propose a conceptually simple analysis of one-shot achievability for processing classical information in quantum systems.
The key point of this work is to demonstrate that the pretty-good measurement directly translates the conditional error probability of a multiple-state discrimination to the error of discriminating a state against the rest. 
This can be viewed as the \emph{one-versus-rest} strategy, and, hence, the pretty-good measurement effectively resembles the quantum union bound in quantum coding design and analysis. 
We obtain an elegant closed-form expression of the average error probability for classical communication over quantum channels with standard random coding and basic properties of the noncommutative minimal.
The proposed method is tight in the sense that it gives tighter one-shot achievability bounds for channels without further constraints such as symmetry (see Section~\ref{sec:comparison}), and it unifies the asymptotic derivations in the large, small, and moderate deviation regimes (Figure~\ref{fig:flow})
Moreover, the analysis naturally applies to various quantum information-theoretic tasks (see Section~\ref{sec:application} and Table~\ref{table:Table}).
This manifests that the proposed method may be considered a fundamental and unified approach to deriving achievable error bounds in quantum information theory. 
In this regard, we may term it as a \emph{one-shot quantum packing lemma} (Theorem~\ref{lemm:packing}).
Essentially, the proposed analysis can be applied to and can sharpen almost all existing results that rely on the Hayashi--Nagaoka operator inequality \cite[Lemma 2]{HN03}; see e.g.~Refs.~\cite{HN03, WR13, DL15, AJW19a, AJW19, Wil17b, QWW18, CHDH-2018, CHDH2-2018}.
The improvement is crucial because every bit counts in a one-shot bound because weaker one-shot bounds could be trivial in certain practical scenarios.
Hence, we expect more applications of the proposed analysis to emerge.
As for the computational aspect, we point out a recently developed quantum algorithm for implementing the pretty-good measurement \cite{GLM+22}.
Last but not least, the proposed achievability analysis also applies to the converse analysis for the covering-type problems \cite{CG22, SGC22a, SGC22b, CG22b}.

We list some open problems along these research directions.
\begin{itemize}
	\item Standard second-order analysis of the achievable coding rate consists of two steps: (i) reducing the underlying task to binary quantum hypothesis testing and (ii) an asymptotic expansion of the quantum hypothesis-testing divergence \cite{TH13, Li14, DPR16, OMW19, PW20}.
	The proposed approach simplifies step (i), and hence, now the bottleneck lies in step (ii).
	Specifically, we conjecture a third-order achievable expansion of the quantum hypothesis-testing divergence in Eq.~\eqref{eq:third_HT}. If Eq.~\eqref{eq:third_HT} holds, then Proposition~\ref{prop:rate} will lead to the best possible third-order coding rate for general classical-quantum channels without further assumptions, i.e.~
	\begin{align} \notag
		\log M \geq n I(\X{\,:\,}\B)_\rho + \sqrt{n V(\X{\,:\,}\B)_\rho } \,   {\Phi}^{-1}(\eps) - O(1).
	\end{align}
	
	\item To the best of our knowledge, conjectures made by Mosonyi and Audenaert \cite[Conjectrue 4.2]{AM14}, and Qi, Wang, and Wilde \cite[Conjecture 18]{QWW18} are still open.
	If they were true, then the established one-shot achievability bound for classical-quantum multiple-access channels (Theorem~\ref{theo:MAC_cq}) will directly imply an upper bound on the error probability by a sum of exponential decays \footnote{Otherwise, one could invoke \cite[Theorem 4.3]{AM14}. However, the resulting error exponents will be weakened by one-half.}.

	\item Can the established strengthened one-shot bound in Eq.~\eqref{eq:tighter} provide a simple proof for the upper bound on the strong converse exponent of classical-quantum channel coding (see Refs.~\cite[Section 5.4]{MO17}, \cite[Proposition IV.5]{MO18}, and \cite[Proposition VI.2.]{CHDH-2018})?
	
	\item In the classical setting (where all the channel output states $\{\rho_{\B}^x\}_{x\in\X}$ mutually commute), the derived bound in Theorem~\ref{theo:main} is still weaker than the \emph{random-coding union bound} proved by Polyanskiy, Poor, and Verd\'u \cite[Theorem 16]{PPV10}, i.e.~the latter implies Theorem~\ref{theo:main} in the commuting case.
	Nevertheless, we remark that Eq.~\eqref{eq:main} can already yield Gallager's random-coding bound in the commuting case. Hence, there are technical noncommutativitiy difficulties that remain to be solved.
	
	\item It is not clear whether the derived bound in Theorem~\ref{theo:main} can lead to Gallager's random-coding exponent \cite{Gal65, Gal68} for general classical-quantum channels.
	It is interesting to see if the recent techniques proposed by Dupuis \cite{Dup21} (see also \cite{CG22}) and Renes \cite{Ren22} can be combined with the proposed analysis.
\end{itemize}

%
%
%
%
%
%
%

\appendix
\section{A Trace Inequality} \label{sec:proof_trace}

This section is devoted to proving the trace inequality in Fact~\ref{fact:lower}:
\begin{align} \label{eq:trace_inequality}
	\Tr[A\wedge B] \geq \Tr\left[ A \frac{B}{A+B} \right].
\end{align}
Note that a special case of Eq.~\eqref{eq:trace_inequality} for $\Tr[A] = \Tr[B] = 1$ is a consequence of Barnum and Knill's theorem \cite{BK02} (see also \cite[Theorem 3.10]{Wat18}).
The result has been extended to the case of general positive semi-definite $A$ and $B$ in the author's previous work \cite[Lemma 3]{SGC22a}.
For completeness, we provide a strengthened proof in the following Lemma~\ref{lemm:trace} that implies the desired Eq.~\eqref{eq:trace_inequality} by extending Sason--Verd{\'u} \cite{SV18} and Renes' idea \cite{Ren17b}.

\begin{lemm}[A trace inequality for noncommutative parallel sum] \label{lemm:trace}
	Let $A$ and $B$ be arbitrary positive semi-definite operators satisfying $\Tr\left[A+B\right]>0$.
	Then, the following holds,
	\begin{align}
		\Tr\left[ A \frac{B}{A+B} \right] &\leq \frac{\Tr\left[ A\vee B \right] \cdot \Tr\left[ A\wedge B \right]}{\Tr[A+B]} \label{eq:tighter_PGM} \\
		\notag
		&\leq \Tr\left[ A\wedge B \right].
	\end{align}
	Here, $A\vee B := \frac{A+B+|A-B|}{2}$ and $A\wedge B := \frac{A+B-|A-B|}{2}$.
\end{lemm}
\begin{remark}
	In the scalar case of positive $a$ and $b$, the inequality $\frac{ab}{a+b}\leq a\wedge b$ is obvious, and the term $\frac{ab}{a+b} = (a^{-1} + b^{-1})^{-1}$ is called the \emph{parallel sum} of $a,b$. Hence, Lemma~\ref{lemm:trace} may be viewed as a noncommutative generalization of it.
	We note that an operator parallel sum $(A^{-1} + B^{-1})^{-1}$ has been studied before (e.g.~\cite[\S 3]{Hia10},\cite[\S 5]{HP14}), and it is related to the Kubo--Ando operator (Harmonic) mean \cite{And79, KA80}.
	Further, it can be shown that $\Tr[(A^{-1} + B^{-1})^{-1}]\leq \Tr[A \frac{B}{A+B}]$, and hence Lemma~\ref{lemm:trace} also provides an upper bound (in trace) to the operator parallel sum.
\end{remark}

\begin{proof}
	We define the \emph{collision divergence} \cite{Ren05} for $A,B\geq 0$ as
	\begin{align} \notag 
		D^*_2(A \,\|\, B) :=  \log \Tr\left[ \left( B^{-\frac{1}{4}} A  B^{-\frac{1}{4}} \right)^2 \right].
	\end{align}
	Let $\{  {\Pi}_A,   {\Pi}_B\}$ be the optimal measurement for distinguishing positive semi-definite operators $A$ and $B$, i.e.~by recalling the Holevo--Helstrom theorem \cite{Hel67, Hol72, Hol76},
	\begin{align} \notag
		\sup_{0\leq T \leq \mathds{1}} \Tr[AT] + \Tr[B(\mathds{1}-T)] = \Tr[A  {\Pi}_A] + \Tr[B   {\Pi}_B] = \Tr[A\vee B].
	\end{align}
	Denote by ``$\oplus$" the direct sum operation. We introduce a measure-and-prepare operation $\Lambda$, which is a completely positive and trace-preserving (CPTP) map:
	\begin{align} \notag
		\Lambda:(\,\cdot\,) \mapsto \Tr\left[ (\, \cdot\,)   {\Pi}_A\oplus   {\Pi}_B \right] \oplus \Tr\left[ (\, \cdot\,) (\mathds{1}-  {\Pi}_A\oplus  {\Pi}_B) \right].
	\end{align}
	We calculate the following:
	\begin{align} \notag
		\Lambda(A\oplus B) &= \Tr\left[ A\vee B \right] \oplus \Tr\left[ A\wedge B \right]; \\
		\Lambda\left( (A+B)^{\oplus 2} \right) &= \Tr\left[ A+B \right] \oplus \Tr\left[ A+B \right].
	\end{align}
	Since the map $\exp\{D_2^*(\,\cdot\,\|\,\cdot\,)\}$ is monotonically decreasing under CPTP maps \cite{FL13, MDS+13, WWY14} , we obtain
	\begin{align} 
		\Tr\left[ A \frac{A}{A+B}\right] + \Tr\left[ B \frac{B}{A+B} \right] 
		&= \exp\left\{ D_2^*\left( A\oplus B \,\|\, (A+B)^{\oplus 2} \right) \right\} \label{eq:trace_LHS} \\
		 \notag
		&\geq \exp\left\{ D_2^*\left( \Lambda(A\oplus B) \,\|\, \Lambda\left((A+B)^{\oplus 2}\right) \right) \right\} \\
		&= \exp\left\{ D_2^*\left( \begin{bmatrix} \Tr\left[ A\vee B \right] & 0 \\ 0 & \Tr\left[ A\wedge B \right] \end{bmatrix} \,\left\|\,  \begin{bmatrix} \Tr\left[ A+ B \right] & 0 \\ 0 & \Tr\left[ A+ B \right] \end{bmatrix} \right.  \right) \right\} \notag \\
		 \notag
		&= \frac{\left(\Tr\left[ A\vee B \right]\right)^2}{\Tr[A+B]} + \frac{\left(\Tr\left[ A\wedge B \right]\right)^2}{\Tr[A+B]}.
	\end{align}
	Noting that the left-hand side of Eq.~\eqref{eq:trace_LHS} can be written as
	\begin{align} \notag
		\Tr\left[ A \frac{A}{A+B}\right] + \Tr\left[ B \frac{B}{A+B} \right] 
		=  \Tr\left[ A+B \right] - 2 \Tr\left[ A \frac{B}{A+B} \right],
	\end{align}
	then the above inequality translates to 
	\begin{align*}
		2 \Tr\left[ A \frac{B}{A+B} \right] &\leq \Tr\left[ A+B \right] - \frac{\left(\Tr\left[ A\vee B \right]\right)^2}{\Tr[A+B]} - \frac{\left(\Tr\left[ A\wedge B \right]\right)^2}{\Tr[A+B]} \\
		&= \frac{ \left(\Tr\left[ A + B \right]\right)^2 - \left(\Tr\left[ A\vee B \right]\right)^2 - \left( \Tr\left[ A\wedge B \right]\right)^2 }{ \Tr\left[ A + B \right]  } \\
		&\overset{\textnormal{(a)}}{=} \frac{ 2 \Tr\left[ A\vee B \right] \cdot \Tr\left[ A\wedge B \right] }{ \Tr\left[ A + B \right]  } \\
		&\overset{\textnormal{(b)}}{\leq} 2 \Tr\left[ A\wedge B \right].
	\end{align*}
	Here, (a) follows from the identity, $A+B = A\vee B + A\wedge B$; 
	and the inequality $\Tr\left[ A\vee B \right] \leq \Tr\left[ A+B \right]$ used in (b) is because that $A+B$ is a feasible solution to the infimum representation of the noncommutative maximum \cite{Hol72, Hol74, Hol76}:
	\begin{align} \notag
		A\vee B = \frac{A+B+|A-B|}{2} = \argmin_{ M = M^\dagger } \left\{ \Tr[M] : M\geq A, \, M \geq B \right\}.
	\end{align}
	Hence, we complete the proof.
\end{proof}

\section*{Acknowledgement}
H.-C.~Cheng would like to thank anonymous reviewers for their detailed suggestions for improving the presentation of the paper.
Jon Tyson for the insightful discussions and pointing out relevant references. H.-C.~Cheng also wants to thank Joseph Renes for the helpful discussions.
This work is done during the visit of Technical University of Munich, Budapest University of Technology and Economics, and Middle East Technical University. H.-C.~Cheng would like to deeply thank Robert K{\"o}nig, Cambyse Rouz{\'e}, Mil{\'a}n Mosonyi, and Bar\i\c{s} Nakibo\u{g}lu for their kind host and hospitality.

H.-C.~Cheng is supported by the Young Scholar Fellowship (Einstein Program) of the Ministry of Science and Technology, Taiwan (R.O.C.) under Grants No.~NSTC 111-2636-E-002-026, No.~NSTC 112-2636-E-002-009, No.~NSTC 112-2119-M-007-006, No.~NSTC 112-2119-M-001-006, No.~NSTC 112-2124-M-002-003, by the Yushan Young Scholar Program of the Ministry of Education, Taiwan (R.O.C.) under Grants No.~NTU-111V1904-3, No.~NTU-112V1904-4, and by the research project ``Pioneering Research in Forefront Quantum Computing, Learning and Engineering'' of National Taiwan University under Grant No. NTC-CC-112L893405. H.-C.~Cheng acknowledges support from the ``Center for Advanced Computing and Imaging in Biomedicine (NTU-112L900702)'' through the Featured Areas Research Center Program within the framework of the Higher Education Sprout Project by the Ministry of Education (MOE) in Taiwan.


{\larger
\bibliographystyle{myIEEEtran}
\bibliography{reference.bib}
}

\end{document}